\documentclass[journal,onecolumn,draftclsnofoot,12pt]{IEEEtran}
\usepackage{bm}
\usepackage{amsmath,amssymb,amsfonts}
\usepackage{algorithmic}
\usepackage{graphicx}
\usepackage{textcomp}
\usepackage{xcolor}
\usepackage{cleveref}
\usepackage[backend=biber,style=ieee]{biblatex}
\usepackage{amsthm}
\usepackage{algorithm}
\usepackage{subfigure}
\usepackage{multirow}
\usepackage{verbatim}
\usepackage{enumitem}
\usepackage{url}
\usepackage{tikz}
\usetikzlibrary{matrix}
\usepackage{nicematrix}
\usepackage{mathtools}
\usepackage{stackengine}
\newtheorem{theorem}{Theorem}
\newtheorem{lemma}{Lemma}
\newtheorem{proposition}{Proposition}

\newtheorem{remark}{Remark}
\newtheorem{example}{Example}
\newtheorem{corollary}{Corollary}
\newtheorem{cor}{Corollary}

\newcommand{\Adn}{\mathbf{A}_{n,d}}

\newcommand{\Ex}{\mathbb{E}}
\newcommand{\Prob}{\mathbb{P}r}

\title{Fundamental Limits of Hypergraph Edge Partitioning under Independent Edge Sampling}

\author{Javad Maheri, K.~K.~Krishnan Namboodiri, and Petros Elia%
\thanks{The authors are with the Communication Systems Department, EURECOM, Sophia Antipolis, France (e-mail: \{maheri, karakkad, elia\}@eurecom.fr).}
\thanks{This work was supported by the France 2030 ANR program ``PEPR Networks of the Future'' (ref. ANR-22-PEFT-0010). 
}}

\begin{document}

\maketitle
\begin{abstract}
Hypergraph edge partitioning is a central problem in theoretical and applied computer science, with broad impact on distributed computation, communications, optimization, and machine learning. In this setting, one is given a collection of hyperedges -- each consisting of up to $d$ vertices from a ground set of size $n$ -- and seeks to assign these hyperedges across $N$ partitions so as to minimize, for example, the vertex footprint, i.e., the maximum number of vertices that appear in any partition.  
We here identify the fundamental limits of hypergraph edge partitioning  -- optimized over all conceivable algorithms -- for a broad class of probabilistic hypergraph models where each hyperedge may appear independently with \emph{its own} probability; a model sufficiently general to encompass well-known models such as the Degree-Corrected or Mixed-Membership models, the Hypergraph Stochastic Block model, the Latent-Space/Geometric or Kernel Models, and others. 
By pairing our deterministic partitioner with a new converse, we first show that, for any $n,d$, and under the very mild condition of $N \leq \binom{\lfloor\sqrt{\frac{nd}{2}}\rfloor}{d}$, as long as the hyperedge set $\mathbf{X}$ satisfies \textcolor{black}{$|\mathbf{X}| \geq \frac{13}{2}nN\left(\log N +\frac{12}{13}  \frac{\log (2n^z)}{n}\right)$}, then with probability at least \textcolor{black}{$1-2/3n^z$}, no algorithm can provide a footprint $\pi_{\mathbf{X}}$ less than $$\pi^{\bigstar}_{\mathbf{X}} = \textcolor{black}{\frac{1}{2\sqrt{2}}}\frac{n}{N^{1/d}}. $$ We then show that our hypergraph partitioner comes to within a small constant factor from $\pi^{\bigstar}_{\mathbf{X}}$, for \emph{each and every hypergraph} $\mathbf{X}$. This optimality captures dense and sparse hypergraphs alike (with sizes down to linear in $n$), and it additionally entails a near-optimally balanced allocation of hyperedges across partitions. The result is not asymptotic, and rather surprisingly, this near-optimal performance can be achieved by a blind and efficient deterministic scheme (the overall complexity is $O(|\mathbf{X}|)$) with \emph{fixed} vertex placement independent of $\mathbf{X}$. This is perhaps the most surprising consequence of our results; that under our assumptions, near-optimal footprint minimization can be universally achieved by a partitioning scheme whose vertex placement is completely independent of the realized hypergraph.
\end{abstract}
\maketitle

\thispagestyle{empty}

\section{Introduction} 

Hypergraph partitioning is a fundamental combinatorial optimization problem that arises in a wide range of applications spanning engineering, scientific computing, machine learning, data analytics, and VLSI design. Classical hypergraph partitioning seeks to divide a hypergraph into balanced blocks while optimizing objectives such as cut size, communication volume, or load balance \cite{KarypisIrregular,aykanat2008multi,BUI1992153,kernighan1970efficient,fischer2009approximate,Chekuri2018MinimumCA,catalyurek2019profile,li2018fiedler}. By employing hypergraphs -- where hyperedges connect arbitrary subsets of vertices -- such formulations can naturally capture higher-order dependencies that arise in complex systems.  

\nocite{schlag2023high}
As a consequence, hypergraph partitioning has indeed been successful in tackling many challenges in a broad range of domains, partly because hypergraphic models provide a natural representation of multi-way interactions in various settings. Beyond high-performance scientific computing, and VLSI design~\cite{MultiLevelVLSI,1676460,devine2006parallel}, hypergraph partitioning formulations are also widely used in distributed databases, large-scale data analytics, social network analysis, and route planning \cite{karypis1999multilevel,heuer2021multilevel,catalyurek2023survey,schlag2023high}, where hypergraph models facilitate scalable storage, computation, and communication. Very pertinent are also its applications in sparse matrix computations, where hypergraph partitioning  captures communication requirements in operations such as sparse matrix--vector multiplication, enabling efficient distribution of data and computation across processors \cite{catalyurek1999hypergraph,devine2006parallel,buluc2016parallel,ballard2016hypergraph,DEVECI201569,catalyurek:hal-00763563}. 

In practice, hypergraph partitioning has been dominated by the multilevel paradigm, consisting of coarsening, initial partitioning, and refinement stages. This philosophy underlies influential systems such as hMETIS, PaToH, and KaHyPar \cite{MultiLevelVLSI,schlag2023high,heuer2021multilevel}. Modern developments include deterministic parallel algorithms such as BlasPart and Mt-KaHyPar, which target reproducibility and scalability on massive instances \cite{tong2025blaspart,gottesburen2022deterministic}, while local-search heuristics, particularly the Fiduccia--Mattheyses (FM) algorithm, remain central to partition refinement through iterative improvement procedures \cite{fiduccia1988linear,schlag2023high}. Similar ideas have also been adopted in distributed computing applications; for example, \cite{song2022hypergraph} proposed an online scheduling framework based on hypergraph partitioning and FM-inspired refinement, while \cite{ronzani2026hypergraph} studied GPU-based multilevel hypergraph partitioning under hard resource constraints, while the increasing scale of modern hypergraphs has motivated the development of numerous parallel partitioning frameworks \cite{1372057,Lars2021,Kabiljo2017,MalekiABP21}.  A nice survey of recent advances can be found here \cite{catalyurek2023survey}.

\paragraph{Hypergraph edge partitioning} In this work, we study the specific problem of \emph{hypergraph edge partitioning}, where the objective is to assign the hyperedges of a hypergraph to \(N\) groups while controlling the replication of vertices induced by this assignment. Unlike the more classical vertex-partitioning formulation, where vertices are distributed across blocks while minimizing cut-related objectives, edge partitioning allocates the hyperedges themselves and seeks to limit various vertex-replication costs. 

Such edge-partitioning formulations arise naturally whenever hyperedges represent tasks, computations, transactions, communication patterns, or data dependencies that must be assigned across distributed resources. Such formulations have emerged in a variety of large-scale systems. In distributed graph analytics, assigning edges rather than vertices has been shown to significantly reduce communication bottlenecks and improve scalability on highly irregular graphs, motivating systems such as \textcolor{black}{PowerGraph and subsequent edge-partitioning frameworks \cite{gonzalez2012powergraph,Zhang2017GraphEdgePartitioning,mayer20202pshighqualityedgepartitioning,Bourse2014}. Similar ideas have been employed in large-scale graph-processing infrastructures operating at the trillion-edge scale, where edge partitioning provides a practical mechanism for balancing computation while controlling vertex replication and communication overhead \cite{TrillionEdges,Dynamic}.} Hypergraph edge partitioning also arises naturally in distributed storage systems, database sharding, parallel data analytics, and distributed task allocation, where hyperedges represent groups of data items or computational tasks that must be co-located while minimizing replication costs. 

From an algorithmic perspective, existing approaches include exact and optimization-based formulations \cite{SATPart,mykhailenko2017simulated}, hybrid approaches~\cite{Mayer2021Hybrid}, and balanced hyperedge-partitioning algorithms for large-scale data applications \cite{HEPart}. 
For example, 
\textcolor{black}{
one recent line of research has focused on scalable edge-partitioning frameworks for distributed graph analytics, where again the objective is to balance edge assignments while minimizing vertex replication and communication overhead. Representative examples include PowerGraph and subsequent scalable edge-partitioning systems, which offer additional evidence that assigning edges rather than vertices can substantially improve the processing of highly irregular networks and massive social graphs \cite{gonzalez2012powergraph,Zhang2017GraphEdgePartitioning,mayer20202pshighqualityedgepartitioning,schlag2018scalable}. Another direction has pursued partitioners based on projective planes~\cite{ProjectivePlane} or highly scalable hash-based partitioners, while 
a substantial body of work has focused on streaming graph edge partitioning, where edges arrive sequentially and must be assigned online while always balancing load and limiting vertex replication. Representative examples can be found here~\cite{Petroni2015HDRF,Mayer2018ADWISE,Hoang2019CUSP}, together with other related replication-aware streaming partitioners \cite{Sajjad2016Boosting}, which trade structural awareness for speed and scalability, making them attractive for dynamic and large-scale settings~\cite{PartitioningTrillionEdges,chhabra_et_alV1,Streaming}. 
}

\paragraph{Formalization}
To formalize the hypergraph edge partitioning problem, we first consider a hypergraph $\mathcal{H}(\mathcal{V}, \mathbf{X})$, where $\mathcal{V}$ denotes the vertex set, and $\mathbf{X} \subseteq \binom{\mathcal{V}}{d}$ denotes the hyperedge set. An edge partitioning scheme, by assigning each hyperedge to one of $N$ groups, results in a partition $\{\mathbf{\Phi}_b\}_{b=1}^N$ of $\mathbf{X}$. Labeling $\mathcal{V}$ by $[n] = \{1,2,\dots,n\}$, each group $b$ then `sees' a vertex footprint 
\begin{equation}
\label{eq: alpha_phi}
    \alpha(\mathbf{\Phi}_b)\triangleq\{i \in [n] \mid \exists \mathbf{e} \in \mathbf{\Phi}_b, \quad i \in \mathbf{e} \}
\end{equation}
which brings to the fore the maximum vertex footprint (MVF)\footnote{
There is a variety of such metrics. Another common variant -- which, after normalization, serves as a lower bound on the MVF -- is the average replication factor (ARF), measuring the average number of groups/partitions to which a vertex is assigned. One advantage of the MVF metric is that it admits a meaningful characterization independently of the balance parameter \(\delta_{\mathbf{X}}\), whereas the ARF depends strongly on it. Furthermore, the MVF captures a worst-case quantity, thus naturally reflecting communication and storage bottlenecks in distributed settings~\cite{ajwani2023open}.
} defined as
\begin{equation}
    \label{eq:pi_X }
    \pi_\mathbf{X} \triangleq \max_{b \in [N]} |\alpha(\mathbf{\Phi}_b)|
\end{equation}
which measures the size of the largest induced vertex set $\alpha(\mathbf{\Phi}_b)$ across all groups. Additionally, the load-balance constraint is captured by the parameter $\delta_\mathbf{X}$, defined, for example, as the ratio between the maximum group size and the ideal group size. 

Despite the aforementioned  algorithmic advances, the theoretical understanding of hypergraph edge partitioning -- whether this is for the above footprint metric, or related metrics -- remains incomplete. While graph partitioning \((d=2)\) enjoys some performance guarantees 
\cite{Dynamic,TrillionEdges,Zhang2017GraphEdgePartitioning}, comparatively little is known about the fundamental limits of hypergraph edge partitioning for general hypergraphs \((d\ge 3)\). The present work provides explicit characterizations of the achievable and unavoidable footprint in a broad class of probabilistic hypergraph models.

\paragraph{The IC-Design}
\textcolor{black}{The recent work in~\cite{IC,maheri2026order} proposed a combinatorial construction, termed the Interweaved-Cliques (IC) design, that leverages carefully structured interweaving sets to allow for hypergraph edge partitioning. Clique-based set structures also arise in several problems in coding and information theory, including coded caching and coded distributed computing~\cite{MaN,MKR,PNR2,Li2018CodedMapReduce,BrEl2,BrEl}. In contrast, the IC design leverages such structures for hypergraph edge partitioning, with the objective of minimizing the vertex footprint.} Focusing on the setting of $d$-uniform hypergraphs, the work was able to show the order-optimality of the design, doing so only for the extreme case of \textit{very dense and homogeneous hypergraphs} -- where $\mathbf{X}$ is drawn uniformly and densely from the different neighborhoods/subsets of the complete $d$-uniform hypergraph $\Adn$.  
From a probabilistic standpoint, the model in~\cite{IC} can be interpreted as capturing a high-entropy, dense regime, corresponding to a worst-case instance of an i.i.d. thinning model where $\mathbf{X}$ samples uniformly and densely from $\Adn$. As such, this prior approach does not account for heterogeneity, structural dependencies, or sparsity in the hyperedge set.

\paragraph{Going beyond maximal entropy hypergraphs -- capturing heterogeneity, structural dependencies, and sparsity} In contrast, the present work considers a substantially more general model in which each potential hyperedge $\mathbf{e} \in \Adn$ is included independently with its own prescribed probability $\varphi_{\mathbf{e}}$, allowing for sparsity, broad heterogeneity and differently-structured neighborhoods in $\mathbf{X}$. This generalization enables us to optimally capture various structure regimes and a very broad sparsity range. In particular, while prior results~\cite{IC} were confined to dense instances where $|\mathbf{X}|$ is within a constant factor from the size $|\Adn|$ of the complete hypergraph $\Adn$, we here establish tight fundamental limits under the much broader condition $|\mathbf{X}| \gtrsim n N \log N$, thereby extending optimality down to very sparse $\mathbf{X}$. For example, under a large-$n$ interpretation, our presented result here extends down to \textit{hypergraphs whose edge set size can be of minimal order, as low as linear in $n$}.

We will treat both the $d$-uniform case and the general case of edge sets
\[
\mathbf{X}\subseteq \mathbf{A}_{n,1}\cup\mathbf{A}_{n,2}\cup\cdots\cup\mathbf{A}_{n,d}
\]
which may contain singleton hyperedges, pairwise hyperedges, and, more generally, hyperedges of size up to $d$.

\paragraph{Positioning this work}
Our work is related to a growing body of theory-oriented research on hypergraph optimization and partitioning, including approximation algorithms for hypergraph partitioning \cite{fischer2009approximate}, structural studies of hypergraph cuts and sparsification \cite{Chekuri2018MinimumCA}, hypergraph-based formulations for combinatorial optimization objectives \cite{catalyurek2019profile}, and spectral approaches to hypergraph partitioning based on Laplacian tensors \cite{li2018fiedler}. 
Complementary to many of the existing works in hypergraph partitioning, our emphasis is less on designing of improved partitioning algorithms, and more on identifying the fundamental limits of hypergraph edge partitioning under broad probabilistic hypergraph models. Specifically, we derive converse bounds on the minimum achievable maximum vertex footprint and show that an explicit, universal and deterministic construction achieves these limits to within small constant factors. 

\nocite{maheri2025constructing}
\section{Problem Definition and Setup} \label{sec:ProbDef}

We begin with the hypergraph model. 

\subsection{Hypergraph Model}
For given positive integers ($n$, $d$, $N$), we consider a hypergraph $\mathcal{H}(\mathcal{V}, \mathbf{X})$, where $\mathcal{V}=[n] \triangleq \{1, 2, \ldots, n\}$ is the vertex set, and where the hyperedge set $\mathbf{X}$ is a subset of the ambient space $\mathbf{A}$. In the $d$-uniform case, the ambient space takes the form $\mathbf{A}=\Adn$, while in the general case, it takes the form of $\mathbf{A}=\bigcup_{k=1}^{d}\mathbf{A}_{n,k}$; the collection of all $k$-element subsets of $[n]$, from $k=1$ to $k=d$. 

An edge partitioning of $\mathbf{X}$ into $N$ groups is a map from $\mathbf{X}$ to a collection $\{\mathbf{\Phi}_1, \mathbf{\Phi}_2, \ldots, \mathbf{\Phi}_N\}$ consisting of $N$ groups that form a partition of $\mathbf{X}$, where each hyperedge of $\mathbf{X}$ is assigned to exactly one group, the groups are pairwise disjoint, and their union equals $\mathbf{X}$. We denote by $\boldsymbol{\mathcal{S}}_{\mathbf{X}}$ a specific edge partitioning scheme that produces such a partition.

For each group $b \in [N]$, $\alpha(\mathbf{\Phi}_b)$ is the set of distinct vertices appearing in at least one edge of $\mathbf{\Phi}_b$, i.e.,
\[
    \alpha(\mathbf{\Phi}_b)=\{i \in [n] \mid \exists \mathbf{e} \in \mathbf{\Phi}_b, \quad i \in \mathbf{e} \}.
\]
This represents the \emph{vertex footprint} of group $b$. 

\subsection{Metrics and Objective}

We evaluate partitioning schemes using two primary metrics that capture different aspects. First we have the aforementioned maximum vertex footprint (MVF),
\[
  \pi_{\mathbf{X}} = \max_{b \in [N]} |\alpha(\mathbf{\Phi}_b)|
\]
which represents the maximum number of vertices that appear (across the edges in) a group $\mathbf{\Phi}_b$, maximized across all groups \(b \in [N]\). Then we have the load balance factor, which is defined as \begin{equation}
\label{eq: delta-def}
    \delta_{\mathbf{X}} \triangleq \frac{\max_{b \in [N]} |\mathbf{\Phi}_b|}{\lceil|\mathbf{X}|/N\rceil}
\end{equation}
and which will here measure the imbalance between the maximum group size and the ideal group size. Perfect balance corresponds to $\delta_{\mathbf{X}} = 1$, while larger values indicate increasing imbalance. 
Our primary objective is to characterize the minimum achievable maximum vertex footprint over all partitioning schemes. 

We present below the statistical model, which comes into the picture only for  the converse. 

\subsection{Statistical Model}
We consider a general probabilistic model for \emph{independently} generating hyperedges from the ambient space $\mathbf{A}$. 
Let us denote the corresponding probabilistic generating vector by
\begin{equation}
 \boldsymbol{\varphi}\triangleq (\varphi_{\mathbf{e}})_{\mathbf{e}\in \mathbf{A}}
\end{equation}
where each hyperedge $\mathbf{e}$ may appear with its own probability $\varphi_{\mathbf{e}}$ that can vary arbitrarily, but which, first for the $d$-uniform case, must satisfy
\begin{equation} \label{eq:phi1}
 \varphi_{\mathbf{e}}\ge \varphi_{\mathrm{min}} = \varphi_{\mathrm{min}}^{u}  \triangleq \min\left\{\frac{15nN(\log N+\frac{4}{5}\frac{\log (2n^z)}{n})}{{n \choose d}},1 \right\}
 \end{equation}
 \textcolor{black}{while for the non-uniform case, it must satisfy 
  \begin{equation} \label{eq:phi2}
 \varphi_{\mathbf{e}}\ge \varphi_{\mathrm{min}}= \varphi_{\mathrm{min}}^{n} \triangleq \min\left\{\frac{13nN (\log N+\frac{12}{13}\frac{\log(2n^z)}{n})}{{\sum_{k=1}^{d}\binom{n}{k}}},1\right\}
\end{equation}} 
to mainly account for the larger size of the non-uniform ambient space. The above parameter $z\geq 1$ calibrates, as we will see, the probability that the derived performance guarantees apply. To facilitate a unified presentation of the results, we are using $\varphi_{\min}$ to mean $\varphi_{\min}^{u}$ in the $d$-uniform setting and to mean $\varphi_{\min}^{n}$ in the non-uniform setting. 

\begin{remark}
The above model entails hypergraphs whose cardinality ranges from the dense or even complete case of $|\mathbf{X}| \approx |\mathbf A|$, to an opposite scenario where edge-set sizes can go down to 
$\varphi_{\min}|\mathbf A|
\approx
nN\log N, 
$
that is linear in $n$. 
\end{remark}

\textcolor{black}{\paragraph{Model generality}
The above formulation does not impose any specific structural form on
the probabilities $\varphi_{\mathbf{e}}$. Instead, the vector
$\boldsymbol{\varphi}$ may represent an arbitrary statistical mechanism
that satisfies the edge-independence assumption, while still inducing a
broad range of structures across the observed hypergraph.
Consequently, many commonly studied hypergraph models are encompassed by
our framework, either directly or conditionally on the realization of
their latent variables\footnote{Models with latent variables (e.g., stochastic block, mixed-membership, latent-space, or geometric models) are encompassed conditionally on the realization of those latent variables, which induces edge-specific probabilities $\varphi_{\mathbf e}$ and independent edge sampling.  In more specific terms: a model, such as a geometric model, may shape arbitrarily the probabilities $\varphi_{\mathbf{e}}$ with which each hyperedge $\mathbf{e} \in \mathbf{A}$ appears, and after that, the sampling is done independently, edge after edge.}. Representative examples include:
\begin{itemize}[leftmargin=10pt]
\item \textbf{Uniform iid model:} $\varphi_{\mathbf{e}}$ is the same for all
$\mathbf{e}$.
\item \textbf{Hypergraph stochastic block model (HSBM):}
$\varphi_{\mathbf{e}}$ depends on the community labels of the vertices in
$\mathbf{e}$, with probabilities determined by block-interaction
parameters.
\item \textbf{Degree-corrected or mixed-membership models:}
$\varphi_{\mathbf{e}}$ depends on vertex-specific parameters or
community-membership mixtures.
\item \textbf{Latent-space/geometric or kernel models:}
conditional on the latent variables associated with the vertices,
$\varphi_{\mathbf{e}}$ is determined by a kernel or similarity function
of those latent variables.
\end{itemize}
}




\subsection{Paper Organization}

After Section~\ref{sec:ProbDef} introduced the hypergraph edge-partitioning problem, the performance metrics of interest, and the probabilistic hypergraph model considered throughout the paper, Section~\ref{sec:MainResults} will present the main theoretical contributions. In this same section, we first establish a converse bound on the minimum achievable maximum vertex footprint through Lemma~\ref{lem:prob_pi}. We then present Theorem~\ref{thm:main}, which characterizes the optimality region of hypergraph edge partitioning under the considered independent edge sampling model. The section further includes Corollary~\ref{cor:fasterZ-cardinality}, which specializes the result to hypergraphs of a given cardinality, as well as Corollaries~\ref{cor:ModifiedConverse} and~\ref{cor:reducedGaps2} which strengthen the main characterization by deriving tighter converse bounds and tighter achievable gaps to optimality. Furthermore, Proposition~\ref{prop: conc on group size} establishes that the proposed construction simultaneously achieves a near-optimal load-balance factor with high probability. Then, Section~\ref{sec:examplesAndRemarks} provides examples, interpretations, and additional discussion that help illustrate the implications of the theoretical results and the resulting gaps to optimality.

Subsequently, Section~\ref{sec:achievability} presents the interweaved-cliques design. We first construct partitions of the complete ambient spaces $\mathbf{A}_{n,d}$ and $\bigcup_{k=1}^{d}\mathbf{A}_{n,k}$, then analyze their footprint properties, extend the construction to the desired number of groups, and finally refine the resulting partition to accommodate an arbitrary realized hyperedge set $\mathbf{X}$.

Finally, Section~\ref{sec:Conclusions} concludes the paper. The proofs of the main results are deferred to Appendix~\ref{appendix:MainProofs}, while additional technical arguments and auxiliary results are provided in Appendix~\ref{appendix:AdditionalProofs}.

We now proceed to present the main results.

\section{Main Results} \label{sec:MainResults}
We here derive the fundamental limits of our edge partitioning problem. We present the results for both the case of $d$-uniform hypergraphs and the general case of non-uniform hypergraphs together; however, the proofs for these cases are developed separately when necessary. Please recall that the ambient space $\mathbf{A}$ will correspond to $\Adn$ and $\bigcup_{k=1}^{d}\mathbf{A}_{n,k}$ in the cases of $d$-uniform and non-uniform hypergraphs, respectively. We now present 
lower and upper bounds on the optimal maximum vertex footprint $\pi^\star_{\mathbf{X}}$,
\begin{equation}
\label{eq: opt regime}
\pi^\star_{\mathbf{X}} \triangleq\min_{\boldsymbol{\mathcal{S}}_{\mathbf{X}}} \left( \max_{b \in [N]} |\alpha(\mathbf{\Phi}_b)| \right)  
\end{equation}
where the optimal (minimum) is taken over all\footnote{By all schemes, we refer to all edge partitioning schemes that assign each edge to exactly one of the $N$ groups.} $N$-edge partitioning schemes \(\boldsymbol{\mathcal{S}}_{\mathbf{X}}\).  

For the following converse, $\mathbf{X}$ is viewed as a realization of an independent random thinning of $\mathbf{A}$, where each edge $\mathbf{e}\in \mathbf{A}$ is retained independently with its own probability $\varphi_{\mathbf{e}}$, satisfying~\eqref{eq:phi1},~\eqref{eq:phi2}. 

Let us now proceed with the converse on the optimal maximum vertex footprint, $\pi_{\mathbf{X}}^{\star}$. 

\begin{lemma}[Converse bound]
\label{lem:prob_pi}
With probability at least $1-\frac{1}{n^z}$, $\mathbf X$ will be such that 
the optimal maximum vertex footprint over all $N$-edge partitioning schemes, satisfies
\[
\pi_{\mathbf X}^{\star}
\ge
\frac{1}{2}\frac{n}{N^{1/d}} \ \ \ \ \ \ \text{and} \ \ \ \ \ \ \ \ \pi_{\mathbf X}^{\star}
\ge
\textcolor{black}{\frac{n}{2\sqrt{2}N^{1/d}}}
\]
in the uniform and non-uniform cases, respectively. 
\end{lemma}
\begin{proof}
    The proof of Lemma~\ref{lem:prob_pi} will treat the cases $\mathbf A=\Adn$ and $\mathbf A=\bigcup_{k=1}^{d}\mathbf A_{n,k}$ separately, the first in Appendix~\ref{appendix: proof of prob_pi}, and the second in Appendix~\ref{appendix: proof of prob_pi_non}. The proof proceeds by contradiction, showing that if there existed a partitioning scheme that offered $\pi_{\mathbf{X}} < \frac{1}{2}\frac{n}{N^{1/d}}$ (or $\pi_{\mathbf{X}} < \frac{n}{2\sqrt{2}N^{1/d}}$ in the non-uniform case), then all groups would need to contain a very large number of edges within a small vertex set, which in turn could only occur with a probability less than $\frac{1}{n^z}$. 
\end{proof}

Let us now view our main theorem, under the same model where each hyperedge $\mathbf{e} \in \mathbf{A}$ is included independently with probability $\varphi_{\mathbf{e}}$ (cf.~\eqref{eq:phi1}, \eqref{eq:phi2}), and where $\mathbf{X}$ denotes the resulting hyperedge set.  We remind the reader that this model captures hypergraphs with various structural properties, and effectively entails cardinalities 
\(
|\mathbf{X}| \gtrsim n N \log N
\)
which can be as low as linear in $n$. Finally, we note that any constructive result entails an assumption of $d\leq n/2$. For the trivial case of $d>n/2$, the fundamental limits are clearly known, to within a factor of $2$ automatically, since naturally, $\pi \in [d,n] \subseteq (\frac{n}{2},n] $.

\begin{theorem}[Optimality region of hypergraph edge partitioning]\label{thm:main}
For any 
$n,d$, $N \leq \binom{\left\lfloor\sqrt{\frac{nd}{2}}\right\rfloor}{d}$, then, with probability at least $1-\frac{1}{n^z}$, $\mathbf{X}$ is such that the optimal footprint, over all hypergraph partitioning schemes, satisfies \(
\pi_{\mathbf X}^{\star} \ge
\frac{1}{2}\frac{n}{N^{1/d}} \)  (uniform case) and \( \pi_{\mathbf X}^{\star}
\ge
\frac{n}{2\sqrt{2}N^{1/d}} 
\) (non-uniform case), while for any $\mathbf{X}$, for $f=\max\left\{r\in\mathbb Z:\binom{r}{d}\leq N\right\}$, the IC design achieves 
\begin{equation} \label{eq:PiIC1}
\pi_{IC}
\le
\begin{cases}
\dfrac{dn}{f},
& \text{if } f\mid n,
\\[8pt]
n-(f-d)\!\left(\left\lfloor\frac{n}{f+d}\right\rfloor+1\right)
\le
2d\!\left(\left\lfloor\frac{n}{f+d}\right\rfloor+1\right),
& \text{otherwise}
\end{cases}
\end{equation}
guaranteeing $\pi_{IC}\leq 4e\frac{n}{N^{1/d}}$, which is at most a factor $8e$ and \textcolor{black}{$8\sqrt{2}e$} from the converse in the uniform and non-uniform cases, respectively. 
\end{theorem}
\begin{proof}
    Lemma~\ref{lem:prob_pi} ensured that $\pi^{\star}_{\mathbf X}\geq \frac{1}{2}\frac{n}{N^{1/d}}$ (uniform) and \textcolor{black}{$\pi^{\star}_{\mathbf X}\geq \frac{n}{2\sqrt{2}N^{1/d}}$} (non-uniform), with probability at least $1-\frac{1}{n^z}$. 
     The achievability proof (upper bound) is the direct outcome of Lemma~\ref{lem: delta, power set} in Section~\ref{subsec: N' to N}     
     which analyzes the scheme as presented in Section~\ref{subsec: uniform} (uniform case) and  Section~\ref{subsec: non uniform} (non-uniform case). 
\end{proof}

In the following, we will denote the derived lower bound as $\pi_{\mathrm{LB}}$, where 
$\pi_{\mathrm{LB}} \triangleq \frac{1}{2}\frac{n}{N^{1/d}}$ and \textcolor{black}{$\pi_{\mathrm{LB}} \triangleq \frac{1}{2\sqrt{2}}\frac{n}{N^{1/d}}$} for the uniform and non-uniform cases respectively.

The following corollary centers around the cardinality of the hypergraph $\mathbf X$, where again each hyperedge $\mathbf e\in\mathbf A$ is included independently with probability $\varphi_{\mathbf e}\geq \varphi_{\min}$. 
The corollary reflects a scenario where one is given a hypergraph $\mathbf{X}$, knows its size, and wonders what is the probability (under the model here) that an algorithm can exist that can yield a lower footprint than $\pi_{\mathrm{LB}}$. 

\begin{corollary}[Optimal given cardinality]
\label{cor:fasterZ-cardinality}
For any $n,d,N \leq \binom{\lfloor\sqrt{\frac{nd}{2}}\rfloor}{d}$ and any given hypergraph $\mathbf X\subseteq \mathbf A$ of size $|\mathbf X|\geq \frac{1}{2}\varphi_{\min}|\mathbf A| = 7.5nN\left(\log N+\frac{4}{5}\frac{\log (2n^z)}{n}\right) \approx nN\log N$, (or $|\mathbf X|\geq 6.5nN\left(\log N+\frac{12}{13}\frac{\log (2n^z)}{n}\right)$ in the non-uniform case), then with probability at least $1-\textcolor{black}{\frac{2}{3n^{z}}}$, the optimal footprint (over any algorithm) satisfies $\pi^\star_{\mathbf X}\in\left[\pi_{\mathrm{LB}},\,\pi_{IC}\right]$, where $\frac{\pi_{\mathrm{IC}}}{\pi_{\mathrm{LB}}} \leq \textcolor{black}{8e}$ and $\frac{\pi_{\mathrm{IC}}}{\pi_{\mathrm{LB}}} \leq \textcolor{black}{8\sqrt{2}e}$ for the uniform and non-uniform cases respectively. 
\end{corollary}

\begin{proof}
\textcolor{black}{The proof can be found in \textcolor{black}{Appendix~\ref{appendix: Proof of cor, card X}, which modifies Lemma~\ref{lem:prob_pi} to reveal that with probability} at least $1-\textcolor{black}{\frac{2}{3n^{z}}}$, $\mathbf X$ must be such that no algorithm can achieve a footprint smaller than the corresponding converse lower bound. The achievability is direct from the main theorem.}
\end{proof}

\subsection{Tighter Bounds} \label{sec:tightboundsSec}
The above results reveal what one could call an achievable  \emph{partitioning gain} $n/\pi_{\mathbf{X}}$, which is shown to be very close to $N^{1/d}$ and which is, with high probability, within a constant factor of at most $8e \approx 22$ from the optimal (\textcolor{black}{$8\sqrt{2}e\approx 31$ for non-uniform hypergraphs}). This constant gap is, in fact, partly due to basic approximations, such as the Stirling approximation, used to obtain concise closed-form expressions. 
The following two corollaries tighten the converse and achievability results, revealing even tighter gaps to optimal. We begin with a tighter converse. 



\begin{corollary}[Tighter converse]
\label{cor:ModifiedConverse}
\textcolor{black}{Suppose that, for every edge $\mathbf e$ and any $z\geq 1$, the parameter
$\varphi_{\mathbf e} \leq 1$ satisfies}
\begin{equation}
\label{eq: cond1}
\varphi_{\mathbf e} \ge\frac{ \left({\frac{12 n N^{1-\frac{1}{d}}\log(e^{d}4N)}{4^{1/d}d}+12N\log(2n^z)}\right)}{{\binom{n}{d}}}   \ \ \ \text{and} \ \ \ \varphi_{\mathbf e} \ge \frac{\left({\frac{12 n N^{1-\frac{1}{d}}\log(e^{d}8N)}{8^{1/d}d}+12N\log(2n^z)}\right)}{\sum_{k=1}^n\binom{n}{k}} 
\end{equation}
for the uniform and non-uniform cases respectively. Then, with probability at least \textcolor{black}{$1-\frac{1}{n^z}$}, $\mathbf X$ is such that the optimal maximum vertex footprint over all $N$-edge partitioning schemes, satisfies $\pi_{\mathbf X}^{\star} \ge \frac{n}{(4N)^{1/d}}$ and \textcolor{black}{$\pi_{\mathbf X}^{\star} \ge \frac{n}{(8N)^{1/d}}$} in the uniform and non-uniform cases, respectively.
\end{corollary}
\begin{proof}
    For the $d$-uniform case, the proof follows by repeating the proof of Lemma~\ref{lem:prob_pi}, now  with $t=\frac{n}{(4N)^{1/d}}$ in place of the value of $t$ used in \eqref{eq:lem2, T}. Specifically, reproducing the steps from \eqref{eq: lem1, union prob,1} to \eqref{eq: non-approx}, will  yield the first expression in 
    \eqref{eq: cond1} which now substitutes the expression in~\eqref{eq: non-approx}. Under this condition, the remainder of the proof proceeds unchanged and yields $\pi_{\mathbf X}^{\star}\ge \frac{n}{(4N)^{1/d}}$ with probability at least \textcolor{black}{$1-\frac{1}{n^z}$.} Similarly, in the non-uniform case, substituting \textcolor{black}{$t=\frac{n}{(8N)^{1/d}}$} in~\eqref{eq:lem2, Tnon}, will yield the second expression in 
    \eqref{eq: cond1}, which then yields that $\pi_{\mathbf X}^{\star}\ge \frac{n}{(8N)^{1/d}}$. 
\end{proof}

We now identify tighter gaps to optimal.

\begin{corollary}[Tighter gaps to optimal] \label{cor:reducedGaps2}
   For  \(
\alpha_d \triangleq \frac{d}{(d!)^{1/d}},
\)
then if $N=\binom{f}{d} \leq \binom{\lfloor\sqrt{\frac{nd}{2}}\rfloor}{d}$, the IC design guarantees, for any $\mathbf{X}$, 
\(
\pi_{IC} \leq \frac{2\alpha_d n}{N^{1/d}} < \frac{2en}{N^{1/d}},
\)
and if in addition $f\mid n$, then 
\(
\pi_{IC} \le  \frac{\alpha_d n}{N^{1/d}}  < \frac{en}{N^{1/d}}, 
\)
which entails a gap to converse of at most $4^{1/d}  \alpha_d < 2 \alpha_d < 2 e <5.5 $ (or \textcolor{black}{$8^{1/d} \alpha_d < 2\sqrt{2} \alpha_d < 2\sqrt{2} e <7.7$ for the non-uniform case}). Both gaps $ 4^{1/d} \alpha_d$ and $8^{1/d} \alpha_d $ converge to $e\approx 2.73$ as $d$ increases.


\end{corollary}
\begin{proof}
Theorem~1 guarantees that the IC design achieves
\(
\pi_{IC}\le \frac{4en}{N^{1/d}}
\)
in general, which is reduced to \(\pi_{IC}\le \frac{2en}{N^{1/d}}\) when
\(N=\binom{f}{d}\), and to
\(\pi_{IC}\le \frac{en}{N^{1/d}}\) when, in addition,
\(f\mid n\) (cf.~\cite{IC}). The constant \(e\) in our expressions is a consequence of using the well-known inequality
\(
\binom{f}{d}\le \left(\frac{ef}{d}\right)^d,
\)
which can though be tightened by using instead
\[
\binom{f}{d}
=
\frac{f(f-1)\cdots(f-d+1)}{d!}
\le
\frac{f^d}{d!}
\]
{which then allows us to} substitute \(e\) with
\(
\alpha_d= \frac{d}{(d!)^{1/d}}, 
\)
(noting that \(\alpha_d<e\) since \(d!>(d/e)^d\)). 
Therefore, for any $\mathbf{X}$, the IC design achieves
\(\pi_{IC}\le 4\alpha_d n/N^{1/d}\) in general,
\(\pi_{IC}\le 2\alpha_d n/N^{1/d}\) when
\(N=\binom{f}{d}\), and
\(\pi_{IC}\le \alpha_d n/N^{1/d}\) when, in addition,
\(f\mid n\). 
\end{proof}

\subsection{Near-Optimal Load Balance Factor} \label{sec:OptDelta}
We proceed with the following lemma, which shows that the load balance factor $\delta_{\mathbf X}$, offered by the IC design, remains small with high probability. The lemma will hold under an independent thinning model where now the probability of edge inclusion is identical across all edges, i.e., where $\varphi_{\mathbf e}=\varphi$ for every edge $\mathbf e$. To facilitate a unified presentation of the result, we will use {$\varphi_{\min,\delta}$}, where $\varphi_{\min,\delta}$ is interpreted as $\varphi_{\min,\delta}^{u} = \min\left\{ \frac{54N\log(4Nn^z)}{\binom{n}{d} - (2^{d+1})N},1\right\}\approx \frac{N\log (4Nn^z)}{\binom{n}{d}}$ (cf.~\eqref{eq:varphi_mindelta_u}) for $d$-uniform hypergraphs, and as {$\varphi_{\min,\delta}^{n} = \min\left\{\frac{54N\log(4Nn^z)}{\sum_{k=1}^{d}\binom{n}{k} - (2^{2d+1})N},1\right\} \approx \frac{N\log (4Nn^z)}{\sum_{k=1}^{d}\binom{n}{k}}$ (cf.~\eqref{eq:varphi_mindelta_n})} for non-uniform hypergraphs. 
\begin{proposition}[Load balance factor]
\label{prop: conc on group size}
Let $(n,d,N)$ satisfy $N\leq \binom{\lfloor\sqrt{\frac{nd}{2}}\rfloor}{d}$. In the uniform case, let $d\leq n/4$; in the non-uniform case, let $d\leq n/16$. If $\varphi_{\mathbf e}=\varphi\geq \varphi_{\min,\delta}$, then, with probability at least \textcolor{black}{$1-\tfrac{1}{n^z}$,} the IC design guarantees $\delta_{\mathbf X}\leq 5$.
\end{proposition}
\begin{proof}
The proof is provided in Appendix~\ref{proofofprop: conc on group size} for non-uniform hypergraphs, and is then briefly specialized to the case of $d$-uniform hypergraphs.
\end{proof}

\subsection{Examples and Remarks} \label{sec:examplesAndRemarks}    
The following remarks and examples aim to help the reader better understand the above results. The examples are in reference to the $d$-uniform setting. 
\begin{remark}[Reduced gaps to optimal] \label{rem:reducedGapOpt}
Theorem~\ref{thm:main} provided a simplified lower bound ($\pi_{\mathbf X}^{\star} \ge
\frac{n}{2N^{1/d}} )$ and a simplified achievability guarantee ($\pi_{IC}\leq 4e\frac{n}{N^{1/d}}$), that imply a gap from optimal that is universally bounded by $8e \approx 22$, while then Corollary~\ref{cor:reducedGaps2} considered again simplified expressions, now under some divisibility assumptions, to show gaps that are bounded approximately by $3$ and $5$ for the uniform and non-uniform cases respectively.  As one would expect, the less simplified achievability expressions in~\eqref{eq:PiIC1} and the less simplified converse in Corollary~\ref{cor:ModifiedConverse}, entail even smaller gaps to optimal, as illustrated below.

\end{remark}

\begin{example}[Reduced gaps to optimal] \label{ex:ReducedGapOptim}
In the case of $n=119, d=3, N=\binom{7}{3} = 35$, we know from the main theorem that $f = 7$, which in turn guarantees $\pi_{IC} \leq  \frac{n}{f} d= \frac{119}{7} 3= 51$, which thus entails a gap to optimal (considering the improved converse in Corollary~\ref{cor:ModifiedConverse}) of at most $ 51 / \bigl(\frac{n}{(4N)^{1/d}}) = 51 / \bigl( \frac{119}{\sqrt[3]{140}})   < 2.25$.  
    Furthermore, when $n=34000, d=2, N=\binom{17}{2} = 136$, we have $f = 17$ and $\pi_{IC} \leq  \frac{n}{f} d= \frac{34000}{17} 2= 4000$ (partitioning gain of $\frac{n}{\pi} = \frac{f}{d} = 8.5$), which entails a gap of at most  $\frac{\pi_{IC}}{\pi_{LB}} = 4000 / \bigl(\frac{1}{2} \frac{34000}{\sqrt{136}})   < 2.85$.  Let's now take an example of a large hypergraph with $n=9\cdot 10^7, d=3, N=4060$, which corresponds to $f=30$ (partitioning gain of $f/d=10$), and which entails a gap to converse of at most \( \frac{\pi_{IC} }{ \pi_{LB} }=\frac{nd/f}{n/(4N)^{1/d}}  = 
\frac{(4N)^{1/d}}{f/d}
\approx
2.53 \).

\end{example}

\begin{remark}[The trillion-edge regime] 
    An important observation is that the regime captured by our assumptions includes graph and hypergraph scales that arise in modern industrial systems. For example, the trillion-edge graph-processing setting discussed in \cite{TrillionEdges}, includes social-network graphs ($d=2$) at Facebook scale, with on the order of \(10^9\) vertices and up to \(10^{12}\) edges, thus having density on the order of \(10^{-6}\) when viewed as subgraphs of \(\binom{[n]}{2}\). For the same \(n\approx 10^9\) and the same $N = 256$ as in \cite{TrillionEdges}, our cardinality condition (as is revealed in~\eqref{eq: non-approx} in a tighter, less approximated form), gives (for a large range of $z$ values)
\[
|\mathbf X|
\ge
\frac{6nN^{1-\frac1d}\log\!\bigl((2e)^dN\bigr)}{d}
+
12N\log(2n^z) \approx 4.28\times10^{11}\] which falls within the trillion-edge regime. Table~\ref{tab:sparse_thresholds_z3} sheds some additional light on this.

\end{remark}

\begin{remark}[Average degree]
Let us recall that the average vertex degree satisfies
\(
\overline{\deg}_{\mathbf X}
=
d\frac{|\mathbf X|}{n},
\)
since each hyperedge contributes exactly \(d\) incidences to the sum of all vertex degrees. Recalling the bound on $|\mathbf{X}|$ from ~\eqref{eq: non-approx}, 
we see  that the average vertex degree that guarantees the converse, is
\(
\overline{\deg}_{\mathbf X}
\ge
6N^{1-\frac1d}\log\!\bigl((2e)^dN\bigr)
+
\frac{12dN\log(2n^z)}{n}, 
\)
which, as $n$ increases, is nicely approximated by 
\(
\overline{\deg}_{\mathbf X}
\gtrsim
6N^{1-\frac1d}\log\!\bigl((2e)^dN\bigr), 
\)
which is a constant independent of $n$. For example, when \(N=10\) and $n$ is (roughly speaking) not small, the converse applies whenever the average vertex degree exceeds \textcolor{black}{\(109\)} for graphs (\(d=2\)), \(206\) for \(3\)-uniform hypergraphs (\(d=3\)), \(307\) for \(4\)-uniform hypergraphs (\(d=4\)), and \(407\) for \(5\)-uniform hypergraphs (\(d=5\)). 

\end{remark}


\begin{table}[t]
\centering
\caption{Sparse-regime thresholds for \(N=256\) and \(z=3\). The converse applies whenever \textcolor{black}{\( |\mathbf X| \ge  |\mathbf X|_{\mathrm{min}}\triangleq\frac{6nN^{1-\frac1d}\log((2e)^dN)}{d} + 12N\log(2n^3)\)}, with corresponding density \(\varphi_{\min}=|\mathbf X|_{\mathrm{min}}/\binom{n}{d}\).}
\label{tab:sparse_thresholds_z3}

\begin{tabular}{c|cccc}
\multicolumn{5}{c}{\textbf{Minimum cardinality \( |\mathbf X|_{\mathrm{min}}\)}}\\
\hline
\(n\) & \(d=2\) & \(d=3\) & \(d=4\) & \(d=5\) \\
\hline
\(10^6\) &
\(4.29\times10^8\) &
\(8.57\times10^8\) &
\(1.18\times10^9\) &
\(1.42\times10^9\)
\\
\(10^7\) &
\(4.29\times10^9\) &
\(8.57\times10^9\) &
\(1.18\times10^{10}\) &
\(1.42\times10^{10}\)
\\
\(10^8\) &
\(4.29\times10^{10}\) &
\(8.57\times10^{10}\) &
\(1.18\times10^{11}\) &
\(1.42\times10^{11}\)
\\
\(10^9\) &
\(4.29\times10^{11}\) &
\(8.57\times10^{11}\) &
\(1.18\times10^{12}\) &
\(1.42\times10^{12}\)
\\
\hline
\end{tabular}

\vspace{4mm}

\begin{tabular}{c|cccc}
\multicolumn{5}{c}{\textbf{Corresponding density \( \varphi_{\min}\)}}\\
\hline
\(n\) & \(d=2\) & \(d=3\) & \(d=4\) & \(d=5\) \\
\hline
\(10^6\) &
\(8.58\times10^{-4}\) &
\(5.14\times10^{-9}\) &
\(2.84\times10^{-14}\) &
\(1.70\times10^{-19}\)
\\
\(10^7\) &
\(8.57\times10^{-5}\) &
\(5.14\times10^{-11}\) &
\(2.84\times10^{-17}\) &
\(1.70\times10^{-23}\)
\\
\(10^8\) &
\(8.57\times10^{-6}\) &
\(5.14\times10^{-13}\) &
\(2.84\times10^{-20}\) &
\(1.70\times10^{-27}\)
\\
\(10^9\) &
\(8.57\times10^{-7}\) &
\(5.14\times10^{-15}\) &
\(2.84\times10^{-23}\) &
\(1.70\times10^{-31}\)
\\
\hline
\end{tabular}

\end{table}

\begin{remark}[On ARF]
We are careful to note that, while our achievability results immediately translate into ARF achievability guarantees since 
$ARF \leq \frac{N}{n} \pi$, it is the case that our converse cannot be seen as a converse on the ARF, simply because a lower bound on a worst-case metric (our $\pi$) does not in general imply a corresponding lower bound on an average-case metric.

In addition, let us note that direct comparisons between algorithms designed to optimize these different objectives, ARF and $\pi$, should be interpreted with care. 
To illustrate this distinction, consider a Hypergraph Stochastic Block Model with \(K=100\) communities and \(n=10^6\) vertices, where only \(1\%\) of the vertices participate in cross-community hyperedges. An ARF-oriented partitioner may place the vast majority of vertices in a single partition while replicating only the crossover vertices. In this case, the ARF would be approximately \(0.99\cdot 1 + 0.01\cdot 2 = 1.01\) (for the case of $d=2$), which appears nearly optimal. However, the partitions responsible for handling the crossover hyperedges may accumulate a disproportionately large number of vertices, resulting in a substantially larger worst-case footprint. Thus, a partitioning strategy may achieve an excellent ARF while simultaneously exhibiting poor MVF performance.

\textcolor{black}{\begin{remark}[On parameter $z$]
 As one can see, for larger values of $n$, $\varphi_{\min}$ (cf.~\eqref{eq:phi1},\eqref{eq:phi2}) 
 is effectively unaffected by $z$, which can thus be chosen large enough to easily render the probabilistic converse -- which holds with probability $1-\frac{1}{n^z}$ -- effectively deterministic. 
\end{remark}}
\end{remark}

\section{Achievable Scheme: Interweaved Cliques Design}
\label{sec:achievability}

We now describe the Interweaved-Cliques design; a deterministic edge-partitioning scheme that first constructs a partition of a complete hyperedge universe $\mathbf{A}$, and then refines this partition to the observed edge set $\mathbf{X}$. 
The design, which will be presented for any \({n,d\leq \frac{n}{2}},N\leq \binom{\lfloor\sqrt{\frac{nd}{2}}\rfloor}{d}\) \footnote{The condition on $N$ is very loose, as it allows $N$ to rise up to approximately $ \left(\frac{n}{2d}\right)^{d/2}$, which can far exceed $n$. In the context of, for example, distributed computing, this would entail many more servers than files.}, will in the end produce $N$ groups $\mathbf{\Phi}_1,\dots,\mathbf{\Phi}_N$, such that
\[
\pi_{\mathbf{X}} =\max_{b \in [N]} |\alpha(\mathbf{\Phi}_b)|, \quad \text{and} \quad \delta_{\mathbf{X}} =\frac{\max_{b \in [N]} |\mathbf{\Phi}_b|}{\lceil|\mathbf{X}|/N\rceil}
\] are minimized. 

The construction will proceed in steps. 
 In particular, in Section~\ref{subsec: uniform}, we will present an $N'$-partition of $\mathbf{A}_{n,d}$ (partition the ambient space into $N'$ groups), for some carefully chosen $N' \leq N$, while in Section~\ref{subsec: non uniform} this $N'$-partition will be extended to the non-uniform ambient space $\bigcup_{k=1}^{d}\mathbf{A}_{n,k}$.
Then in Section~\ref{sec:AnalysisNpPartition}, we will provide, in Lemma~\ref{lem: phin,k, cardinality}, an analysis of important parameters of the above $N'$-partitions. Subsequently, in Section~\ref{subsec: N' to N}, we will extend the partitions of $\mathbf{A}_{n,d}$ and $\bigcup_{k=1}^{d}\mathbf{A}_{n,k}$, from \(N'\) groups to the desired \(N\) groups, and we will present Lemma~\ref{lem: delta, power set} that offers guarantees on $\pi$ and $\delta$ for these $N$-partitions, again still of the ambient space \( \mathbf{A}_{n,d}\) and \(\bigcup_{k=1}^{d} \mathbf{A}_{n,k}\).  Finally, in Section~\ref{sec:PartitionX1}, we will adapt the above $N$-partitions, to account for the actual $\mathbf{X}$. We note here that the construction, for the uniform case, can be seen as a calibrated version of the original design in \cite{IC}, and parts are reproduced here; partly for completeness, and partly to facilitate the presentation of the non-uniform case.


\subsection{$N'$-partition of $\mathbf{A}_{n,d}$}
\label{subsec: uniform}
We first construct a partition $\widetilde {\mathbb{P}}$ of the complete $d$-uniform hyperedge set
$\mathbf{A}_{n,d}$. The construction for this part, follows closely the steps in \cite{IC}, and is based on grouping the vertices into $f$ families of vertices (to be described later on), where
\begin{equation}
\label{eq: f def}
f\triangleq\max\left\{r:\binom{r}{d}\leq N\right\}.
\end{equation}

Each family contains $s$ vertices. Let $g = n-fs$.  If $f\mid n$, then 
\begin{equation}
 \label{eq: s selection}   
 s=\frac{n}{f},
\qquad
g=0,
\qquad
\mathcal{E}=\emptyset
\end{equation}
else 
\begin{equation}
 \label{eq: s selection2}   
s=\left\lfloor\frac{n}{f+d}\right\rfloor+1,
\qquad
g=n-fs
\end{equation}
where in the above, we consider a so-called `excluded set' of $g$ vertices\footnote{Think of these as leftover vertices, a result of parameter divisibility; if for example we have $n=102$ vertices, and $f=20$ vertex families, we will have $g=2$ leftover vertices (vertices $101$ and $102$), and these will be handled later on.}
\begin{equation}
\label{eq: excluded}
   \mathcal{E}\triangleq [n]\setminus[n-g]
=
\{n-g+1,\ldots,n\}. 
\end{equation}
We now partition the $n-g=fs$ non-excluded vertices {$[n-g] $}, into $f$ disjoint families
\begin{equation} \label{eq:Fis1}
  \mathcal{F}_i=\{(i-1)s+1,\ldots,is\},
\qquad i\in[f]
\end{equation}
leaving $g$ vertices out, for now. 

At this point, we will, initially, produce $N'\triangleq{f \choose d}$ intermediate edge groups 
$\widetilde {\mathbf{\Phi}}_\sigma, \sigma\in\mathbf{\Sigma}$, 
where 
\begin{equation}
\label{eq: N'intermediate}
\qquad
\mathbf{\Sigma}\triangleq \{\sigma\subseteq[f]:|\sigma|=d\}.   
\end{equation}
In words, each  $\widetilde {\mathbf{\Phi}}_\sigma$
is labeled by a $d$-tuple from $[f]$. We will `allocate' \textcolor{black}{each group $\widetilde {\mathbf{\Phi}}_\sigma$, with a vertex footprint $\alpha(\widetilde {\mathbf{\Phi}}_\sigma)$ of the form 
\begin{align}
    \alpha(\widetilde {\mathbf{\Phi}}_\sigma) = \bigcup_{i\in\sigma}\mathcal{F}_i
\end{align}
involving the vertices that belong to the families indexed by $\sigma$, 
and this will thus entail} 
\begin{equation} \label{eq:alphaPhiSigma}
  |\alpha(\widetilde {\mathbf{\Phi}}_\sigma)|=\sum_{i\in\sigma}|\mathcal{F}_i|=s\cdot d . 
\end{equation}
We will see soon how the hyperedges are assigned to each $\widetilde {\mathbf{\Phi}}_\sigma$.

At this point, for an edge $\mathbf{e}\in \mathbf{A}_{n,d}$, we define its support family by
\begin{equation}
 \mathcal{B}(\mathbf{e})\triangleq \{j\in[f]:\mathbf{e}\cap \mathcal{F}_j\neq\emptyset\}   
\end{equation}
and we let $\beta_{\mathbf{e}} \triangleq|\mathcal{B}(\mathbf{e})|$, where the above $\mathcal{B}(\mathbf{e})$ is simply the families of vertices that an edge $\mathbf{e}$ intersects with. We first separate the edges according to whether they have full family support or not, thus creating 
\[
\mathbf{A}_{\mathrm{full}}\triangleq \{\mathbf{e}\in \mathbf{A}_{n,d}:\beta_{\mathbf{e}}=d\}
\]
which is the set of edges that each intersects exactly one vertex from each of $d$ distinct families. Then, each edge in $\mathbf{A}_{\mathrm{full}}$, is assigned uniquely to the group indexed by its support
\begin{equation}
\label{eq: phi-full}
\widetilde{\mathbf{\Phi}}^{(\mathrm{full})}_\sigma
\triangleq
\{\mathbf{e}\in \mathbf{A}_{\mathrm{full}}:\mathcal{B}(\mathbf{e})=\sigma\},
\qquad \sigma\in\binom{[f]}{d}.   
\end{equation}
We remind ourselves that this partitions the hyperedges of $\mathbf{A}_{\mathrm{full}}$, i.e., those edges in $\mathbf{A}_{n,d}$ that have full support, and which involve no excluded vertices (no vertices from $\mathcal{E}$).  
This is the core of the design, and it will effectively define most or all the \emph{vertices} associated to each group. 

We now proceed to distribute the remaining edges, i.e., those in 
\begin{equation}
 \mathbf{A}_{\mathrm{rem}}\triangleq \mathbf{A}_{n,d}\setminus \mathbf{A}_{\mathrm{full}}
\end{equation}
into the \(N'={ f \choose d}\) intermediate groups. Please note that now we are allowing hyperedges that may also involve vertices from $\mathcal{E}$.

Any edge \(\mathbf{t} \in \mathbf{A}_{\mathrm{rem}}\) will have an arbitrary number \(m_{\mathbf{t}}=|\mathbf{t}  \cap \mathcal{E}|\) of vertices from the set $\mathcal{E}$, and it will thus have  \(d-m_{\mathbf{t}} = |\mathbf{t}  \cap [n-g]|\) vertices from the \([n] \setminus\mathcal{E}\). 
It is easy to see that 
\begin{equation}
\label{eq: valid range m}
     m_{\mathbf{t}}\in [0, \min\{d,g\}]
\end{equation}
 and thus that \(d-m_{\mathbf{t}} \in [\max\{d-g,0\},d]\). We will henceforth, whenever there is no ambiguity, revert to the simpler notation \(m\) instead of $ m_{\mathbf{t}}$.  
We continue, and for every \(m\in [0, \min\{d,g\}]\), we define the set
 \begin{equation} 
    \label{eq:R_m,beta}
    \mathbf{R}_{m,\beta} \triangleq \left\{ \mathbf{t} \in \mathbf{A}_{\mathrm{rem}} \;\middle|\; |\mathcal{B}(\mathbf{t})| = \beta  ,\ |\mathbf{t}  \cap \mathcal{E}|=m \right\}
    \end{equation}
which describes the edges \( \mathbf{t} \) that intersect exactly \( \beta \) families and also contain \(m\) excluded vertices from $\mathcal{E}$. Notice that \(\beta\) can take values in the range \(\big[\lceil\frac{d-m}{s}\rceil, \min\{d-1,d-m\}\big]\) \footnote{This is because $\beta_{\mathbf{t}}\le d-1$ for any \(\mathbf{t} \in \mathbf{A}_{\mathrm{rem}}\).}. If \(m=d\leq g\), then \(\beta =0\), which means that all the entries of \(\mathbf{t}\) are from \(\mathcal{E}\). Let us now partition \(\mathbf{A}_{\mathrm{rem}}\) as follows
\[\mathbf{A}_{\mathrm{rem}}=\bigcup_{m=0}^{\min\{d, g\}}\bigcup_{\beta=\lceil\frac{d-m}{s}\rceil}^{\min\{d-1,d-m\}} \mathbf{R}_{m,\beta}. \]
    For each $\mathcal{I}\in\binom{[f]}{\beta}$, let us first define 
 \begin{equation}
\mathbf{R}_{\beta,\mathcal{I}}
 \triangleq \big\{\mathbf{t}\in\mathbf{A}_{\mathrm{rem}}\ \mid \mathcal{B}(\mathbf{t})=\mathcal{I},\ 
 \mathbf{t}\in\bigcup_{m=0}^{\min\{d-\beta,g\}}\mathbf{R}_{m,\beta}\big\}
 \end{equation}
 to be the set of all $d$-tuples $ \mathbf{t} \in \mathbf{A}_{\mathrm{rem}}$ that intersect exactly all families in $\mathcal{I}$, where in the above, $\mathcal{B}(\mathbf{t})$ denotes the set of families that $\mathbf{t}$ intersects. Let us now also define 
\(\mathbf{R}_{\beta}\triangleq\bigcup_{\mathcal{I}\in\binom{[f]}{\beta}}\mathbf{R}_{\beta,\mathcal{I}} \subset \mathbf{A}_{\mathrm{rem}} \) to be the set of all {edges in $\mathbf{A}_{\mathrm{rem}}$} that meet exactly $\beta$ families.
Furthermore, directly by applying {the established ranges of parameters} \(m\) (cf.~\eqref{eq: valid range m}), we can conclude that the range of \(\beta \) is \(\beta \in [\beta_{\mathrm{min}}, \beta_{\mathrm{max}}]\), where \(
    \beta_{\min} \;\triangleq\; \lceil \frac{d-\min\{d,g\}}{s} \rceil
            \;=\;\lceil \frac{\max\{0,d-g\}}{s} \rceil \)
and $\beta_{\max} \;\triangleq\; d-1.$
   Our next step involves going through this range of $\beta$. For each \(\beta \in [\beta_{\mathrm{min}}, \beta_{\mathrm{max}}]\), we partition each time the set \(\mathbf{R}_{\beta}\) into \( N'={f \choose d} \) groups. This partitioning is described in detail as follows. 
For every group \(\sigma \in {[f] \choose d}\), we proceed to describe \(\widetilde{\mathbf{\Phi}}^{\mathrm{(rem)}}_{\sigma}\), which will be the set of all \textcolor{black}{$d$-tuples} from $\mathbf{A}_{\mathrm{rem}}$ that will be allotted to the intermediate group labeled by \(\sigma\). For each \(\beta \in [\beta_{\mathrm{min}}, \beta_{\mathrm{max}}]\), the group \(\sigma\) has exactly \(\binom{d}{\beta}\) distinct subsets \(\mathcal{I}\in {[f] \choose \beta}\) of cardinality \(|\mathcal I|=\beta\).
For each such $\mathcal I$, let us consider the set \(\mathbf{R}_{\beta,\mathcal I}\) to be the set of all $t_{\beta}$ distinct \(d\)-tuples from \(\mathbf{R}_{\beta}\) which intersect with \(\mathcal{I}\). For a \(d\)-tuple \(\mathbf{e}\), we say that \(\mathbf{e}\) is \emph{eligible} for assignment to a group \(\widetilde{\mathbf{\Phi}}_\sigma\), if \(\mathcal{B}(\mathbf{e})\subset \sigma\), recalling also that for each \(\mathbf{e}\in \mathbf{R}_{\beta}\), there is one \(\mathcal I \in {[f] \choose \beta}\) such that \(\mathcal{I}= \mathcal{B}(\mathbf{e})\). Let us define the set of eligible groups for each \(\mathcal I \in {[f] \choose \beta}\), with cardinality \(\beta\), as
\begin{equation}
\label{eq: g_beta,I}
    \mathcal{G}_{\beta, \mathcal{I}}\triangleq \{\sigma \in {[f] \choose d}| \ \mathcal{I}\subseteq \sigma\}
    \end{equation}
and let us also note that the number, denoted here as $ m_\beta$, of eligible groups, takes the form
\begin{equation}
\label{eq: m_beta}
  m_\beta \;\triangleq\; |\mathcal{G}_{\beta, \mathcal{I}}|=\binom{f-\beta}{d-\beta}.  
\end{equation}  
Now recall that \(t_\beta=|\mathbf{R}_{\beta, \mathcal I}|\), 
and let us introduce the notation
\begin{equation}
 \label{eq: q_beta}
    q_\beta \;\triangleq\; \Big\lfloor\frac{t_\beta}{m_\beta}\Big\rfloor
\end{equation}
and 
\begin{equation}
 \label{eq: r_beta}
    r_\beta \;\triangleq\; t_\beta- q_\beta  m_\beta, \qquad 0\le r_\beta < m_\beta.
\end{equation}
We now consider \(\mathbf{R}_{\beta,\mathcal I}\) in lexicographic order and denote the ordered list by
\begin{equation}
 \mathbf{R}^{\mathrm{lex}}_{\beta,\mathcal I}
\;\triangleq\;
\big( \mathbf{r}_{\beta,\mathcal I}^{(1)}, \mathbf{r}_{\beta,\mathcal I}^{(2)}, \dots, \mathbf{r}_{\beta,\mathcal I}^{(t_\beta)}\big)   
\end{equation}
where \(\mathbf{r}_{\beta,\mathcal I}^{(\ell)}\) describes the \(\ell\)-th \(d\)-tuple in this lexicographic ordering.
Subsequently, let us fix the lexicographic ordering of the \(m_\beta\) groups in \(\mathcal{G}_{\beta,\mathcal I}\), and let us denote this with 
\begin{equation}
\mathcal{G}_{\beta,\mathcal I}^{\mathrm{lex}}
\;=\;
\big(\sigma_1,\sigma_2,\dots,\sigma_{m_\beta}\big).  
\end{equation}
We now partition \(\mathbf{R}^{\mathrm{lex}}_{\beta,\mathcal I}\) into the \(m_\beta\) groups \(\sigma_1,\dots,\sigma_{m_\beta}\).  
Since \(r_\beta=\lfloor \frac{t_\beta}{m_\beta}\rfloor\), each group receives either \(q_\beta\) or \(q_\beta+1\) \(d\)-tuples, which will mean that, in particular, \(r_\beta\) groups receive \(q_\beta+1\) \(d\)-tuples and the remaining \(m_\beta-r_\beta\) groups receive \(q_\beta\) \(d\)-tuples.
Let us now define the index of \(\sigma\) (given a specific \(\mathcal I \subset \sigma\)) by
\begin{equation}
 J_{\beta,\mathcal I}(\sigma)
\;\triangleq\; \{ j \ : \ \sigma=\sigma_j\text{ for some }j\in [m_\beta], \sigma_j \in \mathcal{G}_{\beta,\mathcal I}^{\mathrm{lex}} \}.   
\end{equation}
 
For \(\mathcal I \subset \sigma\) with \(J_{\beta,\mathcal I}(\sigma)=j\in\{1,\dots,m_\beta\}\), we now define the starting and ending indices of the block allocated to \(\sigma\), by
\begin{align}
s_{j,\beta,\mathcal I} &\;=\; (j-1)q_\beta+ \min(j-1,r_\beta) + 1,\label{eq:start 1}\\[4pt]
e_{j,\beta,\mathcal I} &\;=\; j q_\beta + \min(j,r_\beta)\label{eq:end 2}
\end{align}
which gives us the final set 
\begin{equation}\label{eq:allocated-set,C}
\quad
\mathbf{R}_{\beta,\mathcal I,\sigma}
\;=\;
\big\{\, \mathbf{r}_{\beta,\mathcal I}^{(\ell)} \;:\; \ell = s_{j,\beta,\mathcal I},\; s_{j,\beta,\mathcal I}+1,\; \dots,\; e_{j,\beta,\mathcal I}
\,\big\}
\quad
\end{equation}
of allocated edges to intermediate group \(\sigma\), for this \(\mathcal I\). 
Now we go over all \(\mathcal I\) and $\beta$, to combine the disjoint sets $\mathbf{R}_{\beta, \mathcal{I}, \sigma}$, to get
\begin{equation}
\label{eq: phi_rem}
    \widetilde{\mathbf{\Phi}}_{\sigma}^{\mathrm{(rem)}}=\bigcup_{\beta=\beta_{\mathrm{min}}}^{\beta_{\mathrm{max}}}\bigcup_{\mathcal{I}\subset \sigma}\mathbf{R}_{\beta, \mathcal{I}, \sigma}.
\end{equation}

Finally, for every \(\sigma \in {[f]\choose d}\), 
we combine the disjoint sets $\widetilde{\mathbf{\Phi}}_{\sigma}^{(\mathrm{full})}$ in \eqref{eq: phi-full} and $\widetilde{\mathbf{\Phi}}_{\sigma}^{\mathrm{(rem)}}$ in \eqref{eq: phi_rem}, to get the entire set of hyperedges  
 \begin{equation}
     \label{eq:SubfunctionCase1}
     \widetilde{\mathbf{\Phi}}_{\sigma}=\widetilde{\mathbf{\Phi}}_{\sigma}^{(\mathrm{full})}\cup \widetilde{\mathbf{\Phi}}_{\sigma}^{\mathrm{(rem)}}
 \end{equation}
allocated to group \(\sigma\). 
This gives us the partition 
\begin{equation}
\label{eq: full partition}
   \Adn= \bigcup\nolimits_{\sigma \in {[f] \choose d}}\widetilde{\mathbf{\Phi}}_{\sigma}
\end{equation}
of $\Adn$ into \(N'={f \choose d}\) groups. 
\subsection{$N'$-partition of $\bigcup_{k=1}^{d} \mathbf{A}_{n,k}$}
\label{subsec: non uniform}

We now extend the above, to the case of $\mathbf{A} = \bigcup_{k=1}^{d} \mathbf{A}_{n,k}$. 


The same value of $f$ in \eqref{eq: f def} is used, so that $N'=\binom{f}{d}$. As before, if $f\mid n$, we set $s=n/f$; otherwise, we set
\(
s=\left\lfloor\frac{n}{f+d}\right\rfloor+1,
\)
and define again the set $\mathcal{E}$ of (temporarily) excluded vertices as in~\eqref{eq: excluded}. In either case, each intermediate group is still indexed by some \(
\sigma\in\binom{[f]}{d}\).

For \(k=d\), we follow the steps in Section \ref{subsec: uniform} and partition \(\Adn\) into \(N'\) groups as follows 
\begin{equation}
\label{eq: k=d}
   \Adn=\bigcup_{\sigma \in {[f]\choose d}}\widetilde{\mathbf{\Phi}}_{\sigma}^{d}. 
\end{equation} For the rest $k\in[1,d-1]$, we will eventually partition $\mathbf{A}_{n,k}$ as follows
\begin{equation}
\label{eq: An,k}
 \mathbf{A}_{n,k}
=
\bigcup_{\sigma\in\binom{[f]}{d}} \widetilde{\mathbf{\Phi}}^k_\sigma
\end{equation}
where the allocation will again be based on support families \textcolor{black}{and \(\widetilde{\mathbf{\Phi}}^k_\sigma\) denotes the set of  \(k\)-tuples from \(\mathbf{A}_{n,k} \) allotted to the group \(\sigma\). } In particular, for $\mathbf{e} \in \mathbf{A}_{n,k}$, we have
\[
\mathcal{B}({\mathbf{e}})=\{i\in[f]: \mathbf{e}\cap \mathcal{F}_i\neq\emptyset\},
\qquad
\beta_{\mathbf{e}}=|\mathcal{B}({\mathbf{e}})|
\]
and again, if excluded vertices are present (\(g >0\)), we let $m=|\mathbf{e}\cap \mathcal{E}|$; otherwise $m=0$.
For every \(m\in [0, \min\{k,g\}]\), we define the set
 \begin{equation} 
    \label{eq:R_m,beta}
    \mathbf{Q}_{m,\beta} \triangleq \left\{ \mathbf{e} \in \mathbf{A}_{n,k}\; : \; \beta_{\mathbf{e}} = \beta  ,\ |\mathbf{e}  \cap \mathcal{E}|=m \right\}
    \end{equation}
which describes the set of all \(k\)-tuples \( \mathbf{e} \in \mathbf{A}_{n,k}\) that intersect exactly \( \beta \) families and have \(m\) excluded vertices from $\mathcal{E}$. Notice that \(\beta\) can take values in the range \(\big[\lceil\frac{k-m}{s}\rceil, k-m\big]\), and if \(m=k\leq g\), then \(\beta =0\), which means that all the vertices of \(k\)-tuple \(\mathbf{e}\) are from \(\mathcal{E}\). We seek to partition \(\mathbf{A}_{n,k}\) as follows
\begin{equation} \mathbf{A}_{n,k}=\bigcup_{m=0}^{\min\{k, g\}}\bigcup_{\beta=\lceil\frac{k-m}{s}\rceil}^{k-m} \mathbf{Q}_{m,\beta}.
\end{equation}
Towards this, let us fix \(k\in [1, d-1]\), for any \(\beta \in [\lceil\frac{k}{s}\rceil, k]\) and for each $\mathcal{I}\in\binom{[f]}{\beta}$, let us define 
 \begin{equation}
  \label{eq:Q_beta,I}   
 \mathbf{Q}_{\beta,\mathcal{I}}^{k}
 \triangleq \big\{\mathbf{e}\in\mathbf{A}_{n,k}\ :\ \mathcal{B}(\mathbf{e})=\mathcal{I},\ 
 \mathbf{e}\in\bigcup_{m=0}^{\min\{k-\beta,g\}}\mathbf{Q}_{m,\beta}\big\}
  \end{equation}
 to be the set of all $k$-tuples $ \mathbf{e} \in \mathbf{A}_{n,k}$ that intersect exactly all families in $\mathcal{I}$ (again we recall that $\mathcal{B}(\mathbf{e})$ denotes the set of families that $\mathbf{e}$ intersects). It is relatively easy to see that the cardinality \(t_{\beta}^{k}\triangleq|\mathbf{Q}_{\beta,\mathcal{I}}^{k}|\) remains the same for any $\mathcal{I}\in\binom{[f]}{\beta}$ (fixed \(\beta\) and \(k\)). 
Let us now also define 
 \begin{equation}
 \label{eq:Q_beta}
  \mathbf{Q}_{\beta}^{k}\triangleq\bigcup_{\mathcal{I}\in\binom{[f]}{\beta}}\mathbf{Q}_{\beta,\mathcal{I}}^{k} \subset \mathbf{A}_{n,k}
  \end{equation}
to be the set of all \(k\)-tuples that meet exactly $\beta$ families.
Now, we follow a similar approach as that described in Section \ref{subsec: uniform}; for any $k \in [1, d-1]$ and every group \(\sigma \in {[f] \choose d}\), \textcolor{black}{we describe} \(\widetilde{\mathbf{\Phi}}^{k}_{\sigma}\) -- i.e., the set of all $k$-tuples allotted to the group \(\sigma\) from $\mathbf{A}_{n,k}$ -- as follows
\begin{equation}
\label{eq: phi_ k<d}
    \widetilde{\mathbf{\Phi}}_{\sigma}^{k}=\bigcup_{\beta=\beta_{\mathrm{min}}^{k}}^{\beta_{\mathrm{max}}^{k}}\bigcup_{\mathcal{I}\subset \sigma}\mathbf{Q}_{\beta, \mathcal{I}, \sigma}^{k}
\end{equation}
where \(\mathbf{Q}_{\beta, \mathcal{I}, \sigma}^{k}\) is a subset of \(\mathbf{Q}_{\beta, \mathcal{I}}^{k}\) allotted to \(\sigma\), and where
\begin{equation}
\label{eq:beta_min}
    \beta_{\mathrm{min}}^{k} \;\triangleq\; \left\lceil \frac{k-\min\{k,g\}}{s} \right\rceil
            \;=\; \left\lceil \frac{\max\{0,k-g\}}{s} \right\rceil,
\end{equation}
\begin{equation}
\label{eq:beta_max}
    \beta_{\max}^{k} \;\triangleq\; k.
\end{equation}
\textcolor{black}{At this point, group \(\widetilde{\mathbf{\Phi}}_{\sigma}, \ \sigma \in {[f]\choose d},\) will be formed as follows
\begin{equation}
    \label{eq: group sigma, power set}
    \widetilde{\mathbf{\Phi}}_{\sigma}=\bigcup_{k=1}^{d} \widetilde{\mathbf{\Phi}}_{\sigma}^{k}.
\end{equation}}
Finally, combining \eqref{eq: k=d}, \eqref{eq: An,k}, and \eqref{eq: group sigma, power set}, we have the desired partition
\begin{equation}
    \bigcup_{k=1}^{d}\mathbf{A}_{n,k}=\bigcup_{k=1}^{d}\bigcup_{\sigma \in {[f]\choose d}} \widetilde{\mathbf{\Phi}}_{\sigma}^{k}=\bigcup_{\sigma \in {[f]\choose d}}\bigcup_{k=1}^{d} \widetilde{\mathbf{\Phi}}_{\sigma}^{k}=\bigcup_{\sigma \in {[f]\choose d}} \widetilde{\mathbf{\Phi}}_{\sigma}.
\end{equation}

\subsection{Analysis of the $N'$-partition of $\mathbf{A}_{n,d}$ and of $\bigcup_{k=1}^{d}\mathbf{A}_{n,k}$.} \label{sec:AnalysisNpPartition}
The following lemma quantifies the uniformity of the edge allocation across intermediate groups.
\begin{lemma}[Bounding $|\widetilde{\mathbf{\Phi}}_{\sigma}^{k}|$]
\label{lem: phin,k, cardinality}
The cardinality of $\widetilde{\mathbf{\Phi}}_{\sigma}^{k}, \ k \in [1, d]$, is upper and lower bounded as  
\[
\max_{\sigma \in {[f]\choose d}}|\widetilde{\mathbf{\Phi}}_{\sigma}^{k}| \le \frac{\binom{n}{k}}{\binom{f}{d}} + 2^{d} - d
\]
and as
\[
\min_{\sigma \in {[f]\choose d}}|\widetilde{\mathbf{\Phi}}_{\sigma}^{k}| \ge \frac{\binom{n}{k}}{\binom{f}{d}} -2^{d} + d
\]
and thus, for the \(d\)-uniform case, $|\widetilde{\mathbf{\Phi}}_{\sigma}^{d}|$ is bounded as   
\[
 \frac{\binom{n}{d}}{\binom{f}{d}} -2^{d} + d \le|\widetilde{\mathbf{\Phi}}_{\sigma}^{d}| \le \frac{\binom{n}{d}}{\binom{f}{d}} + 2^{d} - d.
\]
\end{lemma}
{\begin{proof}
    The proof of Lemma \ref{lem: phin,k, cardinality} is provided in Appendix \ref{appendix: proof of phin,k, cardinality}.
\end{proof}}
Now, continuing, we sum the bounds above, to get
\begin{align}
   |\cup_{k=1}^{d} \widetilde{\mathbf{\Phi}}_{\sigma}^{k}|\le \sum_{k=1}^{d} \big( \frac{\binom{n}{k}}{\binom{f}{d}} + 2^{d} - d \big)  &= \frac{\sum_{k=1}^{d} \binom{n}{k}}{\binom{f}{d}} + \sum_{k=1}^{d} (2^{d} - d) \label{eq: sum bounnd1,n}\\
    &\le \frac{\sum_{k=1}^{d} \binom{n}{k}}{\binom{f}{d}} + 2^d d  \label{eq: sum bounnd2,n}
\end{align}
where the transition from~\eqref{eq: sum bounnd1,n} to \eqref{eq: sum bounnd2,n} follows because \(d\le 2^{d}\) for \(d\ge 1\).

In conclusion, the maximum group size for a hypergraph with rank at most $d$, satisfies
\begin{equation}
    \max_{\sigma \in [N']} |\widetilde{\mathbf{\Phi}}_{\sigma}|=\max_{\sigma \in \binom{[f]}{d}} \left\{ |\cup_{k=1}^d \widetilde{\mathbf{\Phi}}_{\sigma}^{k}| \right\} \le \frac{\sum_{k=1}^{d} \binom{n}{k}}{\binom{f}{d}} + 2^{d} d 
\end{equation}
and similarly 
\begin{equation}
    \min_{\sigma \in [N']} |\widetilde{\mathbf{\Phi}}_{\sigma}|=\min_{\sigma \in \binom{[f]}{d}} \left\{ |\cup_{k=1}^d \widetilde{\mathbf{\Phi}}_{\sigma}^{k}| \right\} \ge \frac{\sum_{k=1}^{d} \binom{n}{k}}{\binom{f}{d}} - 2^{d} d  .
\end{equation}
These will be used later on. We now continue to extend the partition \(\widetilde{\mathbb{P}}=\{\widetilde{\mathbf{\Phi}}_{\sigma}: \sigma \in {[f] \choose d}\}\) from \(N'={f\choose d}\) to $N$ groups, for both the uniform and non-uniform cases.

\subsection{Extending the Partition of $\mathbf{A}_{n,d}$ and $\bigcup_{k=1}^{d}\mathbf{A}_{n,k}$, from \(N'\) Groups to \(N\) Groups}
\label{subsec: N' to N}
Starting with the uniform case, let us recall (cf.~\eqref{eq: full partition}) that we have already partitioned \(\mathbf{A}_{n,d}\) into \(N' = \binom{f}{d}\) disjoint\footnote{Note that in the uniform case, when context allows, we will often use the simpler notation $\widetilde{\mathbf{\Phi}}_{\sigma_i}$ to represent $ \widetilde{\mathbf{\Phi}}_{\sigma_i}^{d}$.} groups
\[ \widetilde{\mathbf{\Phi}}_{\sigma_1}^{d},\dots,\widetilde{\mathbf{\Phi}}_{\sigma_{N'}}^{d},
\ \  \sigma_1,\dots,\sigma_{N'}\in\binom{[f]}{d}.
\]
Similarly, for the general non-uniform case, let us recall~\eqref{eq: An,k}, where \(\mathbf{A}_{n,k}\) is partitioned into \(N' = \binom{f}{d}\) disjoint groups \[ \widetilde{\mathbf{\Phi}}_{\sigma_1}^{k},\dots,\widetilde{\mathbf{\Phi}}_{\sigma_{N'}}^{k},
\ \  \sigma_1,\dots,\sigma_{N'}\in\binom{[f]}{d}, \quad \text{for all}   \ k\in [1,d].
\] 
From now on, let us assume, without loss of generality, that the indices \(\sigma_1,\dots,\sigma_{N'} \) are in lexicographic order, and then proceed to simplify notation by writing \(\widetilde{\mathbf{\Phi}}^{k}_j=\widetilde{\mathbf{\Phi}}^{k}_{\sigma_j}\), for any \(j\in [N']\). This simply transitions to a scalar index for each group. For each \(k \in [1,d]\), let us declare \footnote{Note that in the uniform case, when context allows, we will often use the simpler notation  \(\widetilde{\mathbb{P}}\) to represent \( \widetilde{\mathbb{P}}^d\).}
\begin{equation}
 \widetilde{\mathbb{P}}^k\triangleq\{\widetilde{\mathbf{\Phi}}_{1}^{k}, \widetilde{\mathbf{\Phi}}_{2}^{k}, \ldots, \widetilde{\mathbf{\Phi}}_{N'}^{k}\}.   
\end{equation}
Then, for each \(k\), \textcolor{black}{we redistribute the edges of the above $N'$ groups of \(\widetilde{\mathbb{P}}^k\), \textcolor{black}{to now all $N$ groups, as described in the following steps.}} First, let us define the variables
\begin{equation}
\label{eq: q,p,r}
  q\triangleq\big\lfloor\frac{N}{N'}\big\rfloor,\qquad
p\triangleq\big\lceil\frac{N}{N'}\big\rceil,\qquad
r\triangleq N\bmod N'
\end{equation}
noting that \(
N=qN'+r,\) where \(p=q\ \text{if }r=0\), and  \(p=q+1\ \text{if }r>0.\) 
At this point, \textcolor{black}{we proceed with dividing up the edges of each of the first $N'$ groups into different smaller parts}, to be followed by the second step of redistributing some of these smaller parts to fill up the empty $N-N'$ groups.
\paragraph*{Step 1 -- Dividing up  (the edges of) each of the first $N'$ groups}   
For each \(b\in[N']\), we define 
\begin{equation}
s_b\triangleq\begin{cases}
p & \text{if } 1\le b\le r,\\[4pt]
q & \text{if } r<b\le N'
\end{cases}
\end{equation}
and we split each \(\widetilde{\mathbf{\Phi}}_b^{k}\) into \(s_b\) disjoint sub-parts using lexicographic ordering (we go through \(\widetilde{\mathbf{\Phi}}_b^{k}\) and we assign hyperedges, in a round-robin manner, to one of the \(s_b\) disjoint sub-parts, one after the other). This naturally yields sub-parts of equal sizes, plus or minus $1$. We keep track of the exact size of each sub-part. We denote these sub-parts by \(\widetilde{\mathbf{\Phi}}_b^{k,(0)},\widetilde{\mathbf{\Phi}}_b^{k,(1)},\dots,\widetilde{\mathbf{\Phi}}_b^{k,(s_b-1)},\) where \(\widetilde{\mathbf{\Phi}}_b^{k}=\bigcup_{b'=0}^{s_b-1}\widetilde{\mathbf{\Phi}}_b^{k,(b')}. \)
\paragraph*{Step 2 -- Shifting sub-parts to extend to \(N\) groups} 
We then relabel these sub-parts to obtain the desired \(N\) groups.  
We define the new \(N\) groups \(\widetilde{\mathbf{\Phi}}_1^{k},\dots,\widetilde{\mathbf{\Phi}}_N^{k}\), by the indexing rule
\begin{equation}
\label{eq: extenstion}
  \widetilde{\mathbf{\Phi}}_{\,b + b'N'}^{k} \;\triangleq\;\widetilde{\mathbf{\Phi}}_b^{k,(b')},
\ \text{for } b\in[N'],\; b'\in\{0,\dots,s_b-1\}.  
\end{equation}
This now gives us the partition \(\widetilde{\mathbb{P}}^{k}\) for all \(k \in [1,d],\) which takes the form
\begin{equation}
\label{eq: p_tilde}
    \widetilde{\mathbb{P}}^{k}=\{\widetilde{\mathbf{\Phi}}_1^{k},\dots,\widetilde{\mathbf{\Phi}}_N^{k}\}.
\end{equation}
As a result, we get \(\widetilde{\mathbb{P}}\) which has now been extended from \(N'\) designed groups to \(N\) desired groups \(\{\widetilde{\mathbf{\Phi}}_{1}, \widetilde{\mathbf{\Phi}}_{2}, \ldots, \widetilde{\mathbf{\Phi}}_{N}\}\), where
\begin{equation}
\label{eq: phi_b:1tod}
\widetilde{\mathbf{\Phi}}_{b}=\bigcup_{k=1}^{d}\widetilde{\mathbf{\Phi}}_{b}^{k}, \quad b\in [N].  
\end{equation}

To evaluate the uniformity of the group sizes across the partition \(\widetilde{\mathbb{P}}\) of \(\bigcup_{k=1}^{d} \mathbf{A}_{n,k}\), we examine the balance parameter $\delta$, which takes the form 
\begin{equation}
    \delta = \frac{\max_{b \in [N]} |\widetilde{\mathbf{\Phi}}_{b}|}{\lceil\sum_{k=1}^{d} \binom{n}{k} / N\rceil}
\end{equation}
in the general case, and \begin{equation}
    \delta = \frac{\max_{b \in [N]} |\widetilde{\mathbf{\Phi}}_{b}|}{\lceil \binom{n}{d} / N\rceil}
\end{equation}
in the specific case of uniform hypergraphs. This, together with $\pi$, is captured in the following lemma\footnote{Note that we are still at the point where we are partitioning the entire universe of hyperedges. We have not reached $\mathbf{X}$ yet.}.
\begin{lemma}
\label{lem: delta, power set}
\textcolor{black}{When partitioning \(\bigcup_{k=1}^{d} \mathbf{A}_{n,k}\) or \( \mathbf{A}_{n,d}\), for any $n,d\le \frac{n}{2}, N \le {\lfloor \sqrt{\frac{nd}{2}}\rfloor \choose d}$, and for $f=\max\left\{r\in \mathbb Z:\binom{r}{d}\leq N\right\}$, the IC design guarantees that $\pi\leq s\cdot d$ when $f\mid n$, where $s=n/f$. Otherwise, $\pi\leq n-(f-d)s\leq 2s\cdot d$, where $s=\left\lfloor \frac{n}{f+d}\right\rfloor+1$.} 
In the uniform case with the same conditions or in the non-uniform case with the additional condition of {\(d\le \frac{n}{5}\),} the design also gives 
\(
\delta\le 3.
\) 
\end{lemma}
\begin{proof}
Regarding $\pi$, the proof results from the following facts. First, we saw from~\eqref{eq:alphaPhiSigma} that $|\alpha(\widetilde{\mathbf{\Phi}}_\sigma)|=s\cdot d$. When $f\mid n$, the design chooses $s=\frac{n}{f}$, in which case $g=0$. Therefore, by~\eqref{eq:alphaPhiSigma}, we directly obtain $\pi=s d$. On the other hand, when the design chooses $s=\left\lfloor \frac{n}{f+d}\right\rfloor+1$, Appendix~\ref{appendix: bound N}, under the assumption that $N \le {\lfloor \sqrt{\frac{nd}{2}}\rfloor \choose d}$ and $d \le \frac{n}{2}$, yields
\begin{equation}
    \label{eq:pi2sd}
    \pi\leq sd+g = n-(f-d)s\leq 2sd
\end{equation} 
which is further readily bounded, in the uniform case, as 
\(
\pi\le 2sd \le \frac{4e n}{N^{\frac{1}{d}}},
\) {in the proof}~\footnote{Our construction here has been slightly extended, compared to that in \cite{IC}, to now cover all cases of $d\leq n/2$ (rather than $d\leq n/32$ from \cite{IC}). It is a trivial exercise to see that the above bound $\pi\leq sd+g = n-(f-d)s\leq 2sd$ carries over directly to our currently extended range of $d$ (to have a closer look at this, please see Equation~\eqref{eq:gv1} in Appendix~\ref{appendix: bound N}). 
} of \cite[Lemma 6]{IC}. 
Finally, all the above hold for both the uniform and non-uniform cases, simply because the vertex allocation remains the same in both cases. 

The bound on \(\delta\) is handled in Appendix~\ref{appendix: proof of delta, power set} for both the uniform and non-uniform cases.
\end{proof}

\subsection{Partition Refinement for $\mathbf{X}$} \label{sec:PartitionX1}

Let us recall that, up until now, up to Section \ref{subsec: N' to N}, we have only provided a partition \(\widetilde{\mathbb{P}} = \{ \widetilde{\mathbf{\Phi}}_1, \widetilde{\mathbf{\Phi}}_2, \ldots, \widetilde{\mathbf{\Phi}}_N \}\) 
of  \(\bigcup_{k=1}^{d} \mathbf{A}_{n,k}\) (or of \(\mathbf{A}_{n,d}\) in the $d$-uniform case). Now we must account for $\mathbf{X}$. The vertex allocation is automatic; we simply maintain the exact vertex allocation given by the described partition of the complete hypergraph, and thus we maintain the same $\pi$ guarantee that we have seen in Lemma \ref{lem: delta, power set}. 

For the edge allocation, we now convert \(\widetilde{\mathbb{P}}\) into a partition 
\begin{align}\label{eq:RefinedPartitionX1}  \mathbf{X} = \mathbf{\Phi}_{1} \cup \cdots \cup \mathbf{\Phi}_{N}\end{align}
of the specific edge set \(\mathbf{X} \subseteq \bigcup_{k=1}^{d} \mathbf{A}_{n,k}\) (or \(\mathbf{X} \subseteq \mathbf{A}_{n,d}\) in the uniform case), simply by assigning to each group $b\in [N]$, the entries of the group
\begin{equation}
    \label{eq: sampled group size}   \mathbf{\Phi}_{b}=\widetilde{\mathbf{\Phi}}_b \cap \mathbf{X}.
\end{equation}
We must now account for the specific $\delta_{\mathbf{X}}$. To do so, we define 
\begin{equation}
\label{eq: random, mu_max}
\mu_b\triangleq| \mathbf{\Phi}_{b}|, \ b\in[N], \ \ \mu_{\min} \triangleq \min_{b\in[N]}\mu_b, \quad \mu_{\max} \triangleq \max_{b\in[N]}\mu_b
\end{equation}
where we see that \(\mu_b=|\mathbf{\Phi}_b|=|\bigcup_{k=1}^{d}\mathbf{\Phi}_b^{k}|\) (and in the uniform case, \(\mu_b=|\mathbf{\Phi}_b|=|\mathbf{\Phi}_b^{d}|\)).
In both cases, the balance factor is  
\begin{equation}
\label{eq: random, delta_x,n}
 \delta_{\mathbf{X}}  =  \frac{\max_{b} \mu_b}{\lceil|\mathbf{X}|/N\rceil}.
\end{equation}


While the refinement step in~\eqref{eq:RefinedPartitionX1},~\eqref{eq: sampled group size} is simple enough, it runs the risk of yielding a very large $\delta_\mathbf{X}$, depending on $\mathbf{X}$. Bounding $\delta_{\mathbf{X}}$ is handled in 
Proposition~\ref{prop: conc on group size}, which shows that, under an independent thinning model with identical edge inclusion probability $\varphi\geq \varphi_{\min,\delta}$, the IC design guarantees, with probability at least $1-\tfrac{1}{n^z}$, a bounded load balance factor $\delta_{\mathbf X}\leq 5$ for both $d$-uniform and non-uniform hypergraphs. 

At this point, we can also conclude that a notable feature of the proposed IC design is its extremely low computational complexity. The vertex placement is determined once and for all as an explicit deterministic function of $(n,d,N)$ and remains completely independent of the realized hypergraph. Thereafter, each hyperedge is assigned directly to a partition through a fixed deterministic mapping, requiring only the evaluation of a simple combinatorial rule. Consequently, the overall complexity is $O(|\mathbf{X}|)$, which is optimal up to constant factors since every hyperedge must be inspected at least once. 
\section{Conclusions} \label{sec:Conclusions}
Our results yield the following conclusions which hold for the problem of minimizing the MVF in hypergraph \emph{edge} partitioning, under our general  model where each hyperedge $\mathbf{e} \in \mathbf{A}$ may appear independently with its own probability $\varphi_{\mathbf{e}}$, forming sufficiently populated hypergraphs with \(
|\mathbf X| \gtrsim nN\log N.
\) A central conclusion of this work is that, over a broad class of probabilistic hypergraph models, the fundamental limits of hypergraph edge partitioning can be characterized up to explicit and small constant factors;  
Over a broad range of probabilistic hypergraph models, the optimal footprint is -- with a very high probability $1-\frac{1}{n^z}$ -- within a small constant factor of $$\pi^{\bigstar}_{\mathbf{X}} = \frac{1}{2\sqrt{2}}\frac{n}{N^{1/d}}. $$ 

This is surprising, as it says that as long as \(
|\mathbf X| \gtrsim nN\log N.
\), the detailed structure of the hypergraph may not play a dominant role in determining the optimal MVF. While identifying structural features such as communities, clusters, dense regions, local connectivity patterns, and favorable matching opportunities, is still of paramount importance,  our results suggest that, once the hypergraph is sufficiently populated, the optimal footprint is governed to a sizable extent by the combinatorial parameters \((n,d,N)\), with a comparatively limited dependence on the finer structural details of the hypergraph. 

This also entails a strong form of universality. Classical hypergraph partitioning methods are often calibrated to the aforementioned specific structural properties of the input hypergraph. This calibration is most certainly important in practice. What we now know though is that a single, fixed, deterministic construction achieves a near-optimal performance uniformly across a broad range of models. Thus, both the fundamental limit and the achievable scheme exhibit a remarkable degree of universality, again suggesting that -- in our setting, and under our assumptions -- for footprint minimization, the detailed structure of the hypergraph may not carry abundant information, beyond what is already encoded in the underlying combinatorial parameters. 

Our construction -- which achieves a footprint within a constant factor of the optimum, over a broad class of probabilistic models -- follows a simple philosophy: the vertex placement is fixed a priori and depends only on $(n,d,N)$, remaining completely independent of the realized hypergraph. This is perhaps the most surprising outcome of this work; That for the footprint minimization objective considered here, near-optimal performance can be achieved by a placement that depends only on \((n,d,N)\) and is completely independent of the realized hypergraph.

\begin{appendices}
\section{Main Proofs} \label{appendix:MainProofs}
We provide below the various remaining proofs. 

\subsection{Proof of Lemma \ref{lem:prob_pi} - Uniform Case}
\label{appendix: proof of prob_pi}
For  $n, d, N\le \binom{\lfloor\sqrt{\frac{nd}{2}}\rfloor}{d}$, 
consider a randomly drawn \(\mathbf{X}\), 
where each edge \(\mathbf{e} \in \Adn \) is included independently in \(\mathbf{X}\) with its own probability \(\varphi_{\mathbf{e}}\) satisfying condition \eqref{eq:phi1}.  
 First, we consider the general case 
 corresponding to $\varphi_\mathbf{e} \geq \varphi_{\mathrm{min}}^{u}$, where for now,  the threshold probability $\varphi_{\mathrm{min}}^{u}$ is left open, to be derived later on to match the value in~\eqref{eq:phi1}. 
Let \begin{equation}
\label{eq:lem2, T}
    t \triangleq  \frac{1}{2}\frac{n}{N^{1/d}} .
\end{equation}
To establish our bound, consider an arbitrary $N$-edge partitioning scheme $\boldsymbol{\mathcal{S}}_{\mathbf{X}}$ that partitions \(\mathbf{X}\) into $N$ groups \(\mathbf{\Phi}_1\), \(\mathbf{\Phi}_2\), \(\ldots\), \(\mathbf{\Phi}_N\), and let us assume that 
\begin{equation}
|\alpha(\mathbf{\Phi}_b)| < t, \quad \text{for all } b \in [N].
\end{equation}
We will show that such partition can exist with probability at most \(\frac{1}{n^z}\). Let us show this. 

Under our random model, the expected size of \(\mathbf{X}\)  is { \(\sum_{\mathbf{e}\in \Adn}\varphi_{\mathbf{e}} \ge \varphi_{\mathrm{min}}^{u} {n \choose d} \).} 
Consider any scheme $\boldsymbol{\mathcal{S}}_{\mathbf{X}}$ that partitions \(\mathbf{X}\) into $N$ groups \(\mathbf{\Phi}_1\), \(\mathbf{\Phi}_2\), \(\ldots\), \(\mathbf{\Phi}_N\). Naturally, 
at least one group must satisfy $|\mathbf{\Phi}_b|\ge \frac{|\mathbf{X}|}{N}.$ \textcolor{black}{Let us denote \(\mu\triangleq\mathbb{E}[|\mathbf{X}|]\).}
From the Chernoff bound \cite[Theorem~4.5]{mitzenmacher2017probability}, we have 
\begin{equation}
\label{eq: chernoff, card of x}
    \Prob\left( |\mathbf X|\le (1-\varepsilon)\mu\right)\le \exp(-\frac{\mu\varepsilon^2}{2})
\end{equation}
which, for \(\varepsilon=\frac{1}{2}\), yields
\begin{align}
\Prob\left(|\mathbf X|\le \frac{\mu}{2}\right)\le\exp\left(-\frac{1}{2}(\frac{1}{2})^2\mu\right)=\exp\left(-\frac{\mu}{8}\right).  
\end{align}

Since \(\mu=\mathbb{E}[|\mathbf{X}|]\ge \varphi_{\mathrm{min}}^{u} {n \choose d}\), it follows that
\begin{equation}
\label{eq: lem, Eth1}
  \Prob\left(|\mathbf X|\le \frac{\varphi_{\mathrm{min}}^{u} {n \choose d}}{2}\right)\le \Prob\left(|\mathbf X|\le \frac{\mu}{2}\right)\le\exp\left(-\frac{\mu}{8}\right)\le\exp\left(-\frac{\varphi_{\mathrm{min}}^{u} {n \choose d}}{8}\right)
\end{equation}
which, after equating the above with $\frac{1}{2n^z}$ and solving for $\varphi_{\mathrm{min}}^{u}$, yields \begin{equation}\label{eq:phimin-dummy1}
\varphi_{\mathrm{min}}^{u}\ge \frac{8\log(2n^z)}{{n \choose d}}. \end{equation}
Now, for \(\varphi_{\mathrm{min}}^{u}\ge \frac{8\log(2n^z)}{{n \choose d}}\), the RHS of \eqref{eq: lem, Eth1} is less than or equal to \(\frac{1}{2n^z}\), and thus \begin{equation}
\label{eq: x in u}
 \Prob\left(|\mathbf{X}|\ge \frac{1}{2}\mathbb{E}\{|\mathbf{X}|\}\ge \frac{1}{2}\varphi_{\mathrm{min}}^{u}\right)   \geq 1-\frac{1}{2n^z}.   
\end{equation}
We will henceforth thus assume that $|\mathbf{X}|\ge \frac{1}{2}\mathbb{E}\{|\mathbf{X}|\} = \frac{\mu}{2} \geq  \frac{1}{2}\varphi_{\mathrm{min}}^{u} {n \choose d}$. \textcolor{black}{Thus, at least one group must satisfy $|\mathbf{\Phi}_b|\ge \frac{|\mathbf{X}|}{N}\ge \frac{ \varphi_{\mathrm{min}}^{u}  \binom{n}{d}}{2N}$.}
{
Let us denote}
\begin{equation}
\label{eq:lem2, Eth}
    E_{\mathrm{th}} \triangleq \frac{ \varphi_{\mathrm{min}}^{u}  \binom{n}{d}}{2N}.
\end{equation}

In the following we will show that, with sufficiently high probability, none of the groups \(\mathbf{\Phi}_1\), \(\mathbf{\Phi}_2\), \(\ldots\), \(\mathbf{\Phi}_N\) can have {greater than or equal to} $E_{\mathrm{th}}$ edges while having a vertex footprint less than $t$. 

Let $\mathcal{S}_b = \alpha(\mathbf{\Phi}_b)$ be the footprint of group $b$, and let us suppose that $|\mathcal{S}_b| < t$. Let $\mathcal{S} \subseteq [n]$ be such that $\mathcal{S}_b \subseteq \mathcal{S} $, and let $|\mathcal{S}| = t$. 
Now consider 
\begin{equation} \label{eq:XS1}
    \mathbf{X}[\mathcal{S}] \triangleq \binom{\mathcal{S}}{d} \cap \mathbf{X}
\end{equation}
and note that 
$\mathbf{\Phi}_b \subseteq \mathbf{X}[\mathcal{S}]$, and thus that 
$|\mathbf{X}[\mathcal{S}]| \ge |\mathbf{\Phi}_b|$. Therefore, at this point, it suffices to (upper) bound the probability that there exists a subset $\mathcal{S} \subseteq [n]$, with \textcolor{black}{$|\mathcal{S}| = t$}, such that $|\mathbf{X}[\mathcal{S}]| \ge E_{\mathrm{th}}$. For any such $\mathcal{S} \subseteq [n]$ with $|\mathcal{S}| = t$, let us define $Y_\mathcal{S} \triangleq |\mathbf{X}[\mathcal{S}]|$.
For purposes of our bound, it suffices to consider the worst case, where  each potential edge within $\mathbf{X}[\mathcal{S}]$ (this entails the \({t\choose d}\) such potential edges from $\binom{\mathcal{S}}{d}$) is included independently with probability $\varphi_{\mathrm{min}}^{u}$, which translates to having \(Y_{\mathcal S} \) be a sum of \({t \choose d}\) independent Bernoulli random variables, which in turn translates to  \(Y_\mathcal{S} \sim \mathrm{Binomial}\left(\binom{t}{d}, \varphi_{\mathrm{min}}^{u}\right).\) 
Now let \begin{equation}
\mu_\mathcal{S} \triangleq \Ex[Y_\mathcal{S}] = \varphi_{\mathrm{min}}^{u} \binom{t}{d}.
\end{equation}
Next, let us apply the multiplicative Chernoff bound
(see \cite[Corollary~A.1.14]{AlonSpencer2008} and 
\cite[Theorem~4.5]{mitzenmacher2017probability}),
for every $\varepsilon>0$ and for \(C_{\varepsilon} = \min\left\{ (1+\varepsilon)\ln(1+\varepsilon) - \varepsilon,\; \frac{\varepsilon^2}{2} \right\}\), to get
\begin{equation}
\label{eq: lem1, chernof}
  \Prob\!\left(Y_{\mathcal S} \ge (1+\varepsilon)\mu_{\mathcal S}\right)\le e^{-C_{\varepsilon}\cdot \mu_{\mathcal S}}.
\end{equation}
Now, let us choose \(\varepsilon\) such that 
\begin{equation}
\label{eq: lem1,Eth,mu_s}
  E_{\mathrm{th}} = (1+\varepsilon)\mu_\mathcal{S}
 \end{equation} 
which happens when \(\varepsilon = \frac{E_{\mathrm{th}}}{\mu_\mathcal{S}} - 1\), which means that 
\begin{align}
\label{eq: lem1, us,Eth1}
\varepsilon+1=&\frac{E_{\mathrm{th}}}{\mu_\mathcal{S}}=\frac{E_{\mathrm{th}}}{\varphi_{\mathrm{min}}^{u} \binom{t}{d}} = \frac{\frac{ \varphi_{\mathrm{min}}^{u}  \binom{n}{d}}{2N}}{\varphi_{\mathrm{min}}^{u} \binom{t}{d}}
= \frac{1}{2N}\frac{n(n-1) \cdots (n-d+1)}{t(t-1) \cdots (t-d+1) }\ge \frac{1}{2N}\left(\frac{n}{t}\right)^d \\
\label{eq: lem1, us,Eth2}
\ge &\frac{1}{2N} \left(\frac{n}{n/(   2N^{1/d})}\right)^d\notag = 2^{d-1}.
\end{align}
This now means that $\varepsilon \geq 1$ for $d\ge 2$, and for such $\varepsilon \geq 1$, \textcolor{black}{a simple verification can show that} 
\begin{equation}
\label{eq: lem1, chernoof1}
e^{-C_{\varepsilon}}\le 1.2^{-(1+\varepsilon)}
\end{equation}
and thus, after substituting \eqref{eq: lem1,Eth,mu_s} and \eqref{eq: lem1, chernoof1} into the Chernoff bound in~\eqref{eq: lem1, chernof}, we can get \begin{equation}
\label{eq: lem1, prob}
  \Prob\!\left(Y_{\mathcal S} \ge E_{\mathrm{th}}\right)
\le1.2^{-(1+\varepsilon)\mu_{\mathcal S}}=1.2^{-E_{\mathrm{th}}}=e^{-\ln(1.2)E_{\mathrm{th}}}\le e^{-E_{\mathrm{th}}/6}
\end{equation}
which applies to any subset \(\mathcal{S}\subset [n]\) with \(|\mathcal S|=t\).  
We now apply the union bound over all such sets $\mathcal{S}$, to get
\begin{align} \label{eq: lem1, union prob}
\Prob\left(\exists \mathcal{S} : |\mathcal{S}| = t, Y_\mathcal{S} \ge E_{\mathrm{th}}\right) \notag  \le \binom{n}{t} \Prob(Y_\mathcal{S} \ge E_{\mathrm{th}}) 
\le \left(\frac{en}{t}\right)^t  e^{-E_{\mathrm{th}}/6}.
\end{align}
Now we can see that for the above to be less than or equal to $\frac{1}{2n^{z}}$ for some $z\geq 1$, we need
\begin{equation}
\label{eq: lem1, union prob,1}
    \left(\frac{en}{t}\right)^t  e^{-E_{\mathrm{th}}/6}\le \frac{1}{2n^z} 
\end{equation}
and after taking the \(\log\) of both sides of the inequality, we get
\begin{equation}
t \log\left(\frac{en}{t}\right) - \frac{E_{\mathrm{th}}}{6} \le -\log(2n^z).
\end{equation}
Substituting $t$ and $E_{\mathrm{th}}$, we obtain the threshold condition on \(\varphi_{\mathrm{min}}^{u}\), in the form
\begin{align}
\frac{n}{2N^{1/d}}  \log\left(\frac{en}{\frac{n}{2N^{1/d}}}\right) - \textcolor{black}{\varphi_{\mathrm{min}}^{u} \frac{{n \choose d}}{12 N}} \le -\log (2n^z) \\
 \varphi_{\mathrm{min}}^{u}  {n \choose d}\ge \frac{6 \ n N  \log\left(2e  \ N^{1/d}\right)}{N^{\frac{1}{d}}}+12N  \log (2n^z)\\ \implies
\textcolor{black}{ \varphi_{\mathrm{min}}^{u}  {n \choose d}\ge \frac{6\ n N ^{1-\frac{1}{d}} \log\left((2e)^{d} \ N\right)}{{ d}}+12 N  \log (2n^z)\label{eq: non-approx}.}
\end{align}
For \(N\ge 2\) and \(d\ge 2\), we can readily see in Lemma \ref{lem:K2} that \(\frac{6 n\ N ^{1-\frac{1}{d}} \log\left((2e)^{d} \ N\right)}{{ d}}\le  15 n N  \log N\), which means that it is sufficient to have 
\begin{align} \label{eq:phiUnif123}
 \varphi_{\mathrm{min}}^{u}  {n \choose d}\ge 15 n N  \log N+12 N  \log(2n^z)\\
  \varphi_{\mathrm{min}}^{u} \ge \frac{15 n  N  (\log N + \frac{12}{15}\frac{\log (2n^z)}{n})}{{{n \choose d}}}.
\end{align}
At this point, we can quickly compare this expression with the bound in~\eqref{eq:phimin-dummy1} to see that 
$\frac{15 n  N  (\log N+\frac{4}{5}\frac{\log (2n^z)}{n})}{{{n \choose d}}}\ge \frac{8\log(2n^z)}{{n \choose d}}$ which allows us to ignore~\eqref{eq:phimin-dummy1}, in the sense that the condition in~\eqref{eq:phimin-dummy1} will be automatically guaranteed once the condition in~\eqref{eq:phiUnif123} is applied. Thus, we can conclude that 
when \textcolor{black}{$\varphi_{\mathbf{e}} \ge \frac{15 n  N  (\log N+\frac{4}{5}\frac{\log (2n^z)}{n})}{{{n \choose d}}}$}, then with probability at least $1 - \frac{1}{n^z}$, $\mathbf{X}$ will be such that no partitioning scheme can give $\pi_{\mathbf{X}}<t=\frac{ n}{2 N^{1/d}}$, since
\begin{align}
\label{eq: lem1,eq1}
\Prob\left[ \{  \pi_{\mathbf{X}}^{\star}<t\}\cap\{|\mathbf{X}|\ge N E_{\mathrm{th}}\}\right]\le \Prob\left[\exists \mathcal{S} : |\mathcal{S}| = t, Y_\mathcal{S} \ge E_{\mathrm{th}}\right]\\
  \label{eq: lem1,eq2}
   \implies\Prob\left[ \{  \pi_{\mathbf{X}}^{\star}<t\}\right]\le \Prob\left[\exists \mathcal{S} : |\mathcal{S}| = t, Y_\mathcal{S} \ge E_{\mathrm{th}}\right]+\Prob\left[\{|\mathbf{X}|< NE_{\mathrm{th}}\}\right]\\
  \implies \Prob\left[ \{  \pi_{\mathbf{X}}^{\star}<t\}\right]\le \frac{1}{n^z}\label{eq:final bound}
\end{align}
\noindent where the step from~\eqref{eq: lem1,eq1} to \eqref{eq: lem1,eq2} follows from the property that \(\mathcal{A}\cap\mathcal{B}\subset \mathcal{C}\) implies \(\mathcal{A}\subset \mathcal{C} \cup\bar{\mathcal{B}}\), while the step from~\eqref{eq: lem1,eq2} to \eqref{eq:final bound} follows after using \eqref{eq: lem, Eth1} and \eqref{eq: lem1, union prob,1}. 
This completes the proof of Lemma~\ref{lem:prob_pi} for the uniform case.\hfill $\square$


\subsection{Proof of Lemma \ref{lem:prob_pi}  - Non-uniform Case}
\label{appendix: proof of prob_pi_non}

Given \(n\), \(d\), and \(N\), we will show, as before, that under the stochastic model with \(\varphi_{\mathbf{e}}\ge \varphi_{\mathrm{min}}^{n}\), 
then with probability at least \(1-\frac{1}{n^z}\), the hypergraph \(\mathbf{X}\) satisfies \(\pi_{\mathbf{X}}^{\star}\ge t \), where now
\begin{equation}\label{eq:lem2, Tnon}
t \triangleq \frac{n}{2\sqrt{2}N^{1/d}}.
\end{equation}
For now, the value of $\varphi_{\mathrm{min}}^{n}$ remains open, and it will be derived later on to match the expression in~\eqref{eq:phi2}.

The proof is similar to that of the case of uniform hypergraphs, in the previous lemma. As before, we proceed by contradiction, by first assuming that there exists a partition \(\mathbf{\Phi}_1\), \(\mathbf{\Phi}_2\), \(\ldots\), \(\mathbf{\Phi}_N\) of \(\mathbf{X}\) with 
\begin{equation}
|\alpha(\mathbf{\Phi}_b)| < t, \quad \text{for all } b \in [N].
\end{equation}
We again consider $\mathcal{S}_b = \alpha(\mathbf{\Phi}_b)$,  $|\mathcal{S}_b| < t$, and a vertex set $\mathcal{S} \supset \mathcal{S}_b$, with $|\mathcal{S}| = t$. 
Then, as in~\eqref{eq:XS1}, we will again consider a (now modified) set \[
    \mathbf{X}[\mathcal{S}] =\left(\bigcup_{k=1}^d \binom{\mathcal{S}}{k}\right) \cap \mathbf{X}\] \noindent 
    and as before, we will focus on bounding the probability that there exists a subset $\mathcal{S} \subseteq [n]$, with $|\mathcal{S}| = t$, such that $Y_\mathcal{S} \triangleq|\mathbf{X}[\mathcal{S}]| \ge E_{\mathrm{th}}$, where now we slightly modify $E_{\mathrm{th}}$ to be 
 \begin{equation}
\label{eq:lem2, Eth}
    E_{\mathrm{th}} \triangleq \frac{ \varphi_{\mathrm{min}}^{n}  \sum_{k=1}^{d}\binom{n}{k}}{2N}.
\end{equation}   
Then we will proceed with a (now slightly modified) distribution 
\(
Y_\mathcal{S} \sim \mathrm{Binomial}\left(\sum_{k=1}^{d}{t\choose k}, \varphi_{\mathrm{min}}^{n}\right), 
\)
whose mean is
\begin{equation}
\mu_\mathcal{S} \triangleq \Ex[Y_\mathcal{S}] = \varphi_{\mathrm{min}}^{n} \sum_{k=1}^{d}{t\choose k}.
\end{equation}
We will again apply the Chernoff bound, which, for every $\varepsilon>0$ and for \\ \(C_{\varepsilon} = \min\left\{ (1+\varepsilon)\ln(1+\varepsilon) - \varepsilon,\; \frac{\varepsilon^2}{2} \right\}\), will again give
\begin{equation}
\label{eq: lem2, chernof}
  \Prob\!\left(Y_{\mathcal S} \ge (1+\varepsilon)\mu_{\mathcal S}\right)\le e^{-C_{\varepsilon}  \mu_{\mathcal S}}
\end{equation}
and again we will choose \(\varepsilon = \frac{E_{\mathrm{th}}}{\mu_{\mathcal S}} - 1\). At this point, there are a few more minor differences, where in particular, the bound in \eqref{eq: lem1, us,Eth1} changes to
\begin{align}
\label{eq: lem2, us,Eth1}
\varepsilon+1=&\frac{E_{\mathrm{th}}}{\mu_\mathcal{S}} =\frac{1}{2N}\frac{\varphi_{\mathrm{min}}^{n}}{\varphi_{\mathrm{min}}^{n}}\frac{\sum_{k=1}^{d}{n\choose k}}{\sum_{k=1}^{d}{t\choose k}} \ge\frac{1}{2N}\frac{\sum_{k=1}^{d}{n \choose k}}{d {t \choose d} }\ge\frac{1}{2N}\frac{{n \choose d}}{d {t \choose d} }\\
\label{eq: lem2, us,Eth2}
\ge &\frac{1}{2Nd}\left(\frac{n}{t}\right)^d=\frac{1}{2N}\frac{(2 \sqrt{2})^{d} N}{d}\ge \frac{(8)^{\frac{d}{2}}}{2d}\ge 2
\end{align}
which again tells us that \(\varepsilon \ge 1\), in which case we know that \(
e^{-C_{\varepsilon}}\le 1.2^{-(1+\varepsilon)}\) which goes into the Chernoff bound (cf.~\eqref{eq: lem2, chernof}) to again give 
\begin{equation}
\label{eq: lem2, prob}
  \Prob\!\left(Y_{\mathcal S} \ge E_{\mathrm{th}}\right)
\le e^{-(\log1.2)(1+\varepsilon)\mu_{\mathcal S}}\le e^{-E_{\mathrm{th}}/6}.
\end{equation}
This in turn says that 
\begin{align}
\Prob\left(\exists \mathcal{S} : |\mathcal{S}| = t, Y_\mathcal{S} \ge E_{\mathrm{th}}\right) 
 \le \binom{n}{t} \Prob(Y_\mathcal{S} \ge E_{\mathrm{th}}) \label{eq: th3, union prob}
 \le \left(\frac{en}{t}\right)^t e^{-E_{\mathrm{th}}/6}
\end{align}
for which we wish to satisfy 
\begin{equation}
    \left(\frac{en}{t}\right)^t e^{-E_{\mathrm{th}}/6}\le \frac{1}{2n^z}.
\end{equation}
We then again take the \(\log\) on both sides of the inequality, to get $t \log\left(\frac{en}{t}\right) - \frac{E_{\mathrm{th}}}{6} \le -\log (2n^z)$, and after substituting for $t$ and $E_{\mathrm{th}}$, we get 
\begin{align} \label{eq:phi-nonsimplifiedNon}
\frac{n}{2\sqrt{2}N^{1/d}}  \log\left(\frac{en}{\frac{n}{2\sqrt{2}N^{1/d}}}\right) - \varphi_{\mathrm{min}}^{n}  \frac{\sum_{k=1}^{d}\binom{n}{k}}{12N} \le -\log (2n^z) \nonumber \\ \implies 
\varphi_{\mathrm{min}}^{n}  \sum_{k=1}^{d}\binom{n}{k}\ge \frac{6n N ^{1-\frac{1}{d}} \log\left((2\sqrt{2}e)^{d}  N\right)}{{ \sqrt{2}d}}+12N  \log (2n^z).
\end{align}
\noindent From Lemma \ref{lem:K1}, we know that 
\(\frac{6 n N ^{1-\frac{1}{d}}\log\left((2\sqrt{2}e)^{d} N\right)}{{ \sqrt{2}d}}\le  13 n  N  \log N,\) 
and thus we conclude that  
\begin{align}
\varphi_{\mathrm{min}}^{n} \geq \frac{13nN (\log N+\frac{12}{13}\frac{\log (2n^z) }{n})}{{\sum_{k=1}^{d}\binom{n}{k}}}.
\end{align}
Thus, with similar arguments as in the uniform case, we can now proceed to conclude that when $\varphi_{\mathbf{e}} \geq \min\left\{\frac{13nN (\log N +\frac{12}{13}\frac{\log (2n^z) }{n})}{{\sum_{k=1}^{d}\binom{n}{k}}},1\right\}$, then with probability at least $1 - \frac{1}{n^z}$, $\mathbf{X}$ will be such that no partitioning scheme can give $\pi_{\mathbf{X}}<t=\frac{ n}{ 2\sqrt{2}N^{1/d}}$.  
This completes the proof of Lemma~\ref{lem:prob_pi} for the non-uniform case as well.\hfill $\square$


\subsection{Proof of Proposition~\ref{prop: conc on group size}}
\label{proofofprop: conc on group size}
We first present the proof for non-uniform hypergraphs, and then explain how it specializes to uniform ones. Let us recall the partition \(\widetilde{\mathbb{P}}=\{\widetilde{\mathbf{\Phi}}_1, \widetilde{\mathbf{\Phi}}_2, \ldots, \widetilde{\mathbf{\Phi}}_N\}\) of \(\bigcup_{k=1}^{d}\mathbf{A}_{n,k}\) seen in Section \ref{subsec: N' to N}, and let us first define
\begin{equation}
\label{eq: random, m_min, m_max}
m_b\triangleq|\widetilde{\mathbf{\Phi}}_b|, \ b\in[N], \ \ \ \     m_{\min} \triangleq \min_{b\in[N]}m_b, \quad m_{\max} \triangleq \max_{b\in[N]}m_b.
\end{equation}
Let us also recall that the term 
\begin{equation}
\label{eq: balance parameter}
 \delta= \frac{m_{\mathrm{max}}}{\lceil\sum_{k=1}^{d}{n \choose k}/N\rceil}
\end{equation}
is shown in Lemma \ref{lem: delta, power set} to be bounded as \(\delta\le 3\) whenever \(d\le \frac{n}{5}\). To avoid confusion, note that this is the guaranteed balance factor, had we been partitioning the complete hypergraph \(\cup_{k=1}^{d}\mathbf{A}_{n,k}\). Now we account for $\mathbf{X}$, and derive how this achievable balance factor may change. Thus, given now \(\mathbf{X}\) and its corresponding partition into $\mathbf{\Phi}_1,\dots,\mathbf{\Phi}_N$, where $\mathbf{\Phi}_b=\widetilde{\mathbf{\Phi}}_b\cap \mathbf{X}, \ b\in[N]$, we need to now upper bound the corresponding 
\begin{equation}
 \delta_{\mathbf{X}} =  \frac{\max_{b\in[N]} |\mathbf{\Phi}_b|}{\lceil{|\mathbf{X}|}/N\rceil}.   
\end{equation}
Towards this, let us first define, for each edge $\mathbf{e} \in \widetilde{\mathbf{\Phi}}_b$, the indicator random variable $I_\mathbf{e} $ where $I_\mathbf{e} = 1$ when $\mathbf{e} \in \mathbf{X}$, and $I_\mathbf{e} = 0$ otherwise. With this, and recalling our assumption that the edges of $\mathbf{X}$ are each independently drawn from  $\bigcup_{k=1}^{d}\mathbf{A}_{n,k}$, we can see that 
\begin{equation}
\label{eq: lem 11, 1}
  \mu_b = |\mathbf{\Phi}_b|= \sum_{\mathbf{e}  \in \widetilde{\mathbf{\Phi}}_b} I_\mathbf{e}=\sum_{\mathbf{e}  \in \bigcup_{k=1}^{d}\widetilde{\mathbf{\Phi}}_b^{k}} I_\mathbf{e}
\end{equation}
where each $I_\mathbf{e}$ follows a Bernoulli distribution with some probability $\varphi_{\mathbf{e}}=\varphi = \mathbb{E}[I_\mathbf{e}] $. 
This in turn means that 
\begin{equation}
\label{eq: lem 11,2}
   \mathbb{E}[\mu_b] = \mathbb{E}\big[\sum_{ \mathbf{e} \in \widetilde{\mathbf{\Phi}}_b} I_\mathbf{e}\big] = \sum_{\mathbf{e} \in \widetilde{\mathbf{\Phi}}_b} \mathbb{E}[I_\mathbf{e}] = \sum_{\mathbf{e} \in \widetilde{\mathbf{\Phi}}_b} \varphi_{\mathbf{e}}=\sum_{\mathbf{e} \in \widetilde{\mathbf{\Phi}}_b} \varphi=\varphi\ |\widetilde{\mathbf{\Phi}}_b|=\varphi\ m_b. 
\end{equation}
Using \cite[Corollary~4.6]{mitzenmacher2017probability}, for \(0 < \varepsilon < 1\) and \(b \in [N]\), we have
\begin{align}
\Prob\big[|\mu_b-\mathbb{E}(\mu_b)| \ge\varepsilon\mathbb{E}(\mu_b)\big] \le 2\exp\!\big(-\tfrac{\varepsilon^2\mathbb{E}(\mu_b)}{3}\big)
\end{align}
and then, by applying the union bound over all \(N\) groups, we get
\begin{align}
 \Prob\big[\exists b\in [N]\mid \mu_b\not\in[(1-\varepsilon)\mathbb{E}(\mu_b),(1+\varepsilon)\mathbb{E}(\mu_b)]\big]  \leq 2 N\max_{b \in [N]}\big(\exp( - \frac{\varepsilon^2\mathbb{E}(\mu_b)}{3} )\big)\\
 \label{eq: lem 11, 5}
  \leq 2 N\exp\big( - \frac{\varepsilon^2}{3}\min\limits_{b \in [N]}\mathbb{E}(\mu_b)\big ).
\end{align}
We also recall that \(\min\limits_{b \in [N]}\mathbb{E}[ \mu_{b} ]= \varphi \ m_{\mathrm{min}}\), because \(\varphi_{\mathbf e}= \varphi\).
Let us now choose \(\varepsilon = \frac{1}{3}\) and \(\eta = \frac{1}{2n^z}\), and substitute these values into \eqref{eq: lem 11, 5}. By solving the following inequality 
\[
2\ N\exp( - \frac{\varphi\ m_{\mathrm{min}}}{27})\le \eta=\frac{1}{2n^z}
\]
we can conclude that under the condition
\begin{equation}
   \label{eq:mu-condition 4b}
\varphi \  m_{\mathrm{min}}\;\ge\; 27\big(\log(2N)+\log(2n^z)\big) 
\end{equation}
it is the case that 
\begin{align}
\label{eq: lem 11, 6}
     \Prob\Big[\exists b\in[N]\mid \mu_b\not\in[\frac{2}{3}\mathbb{E}(\mu_b),\frac{4}{3}\mathbb{E}(\mu_b)]\Big] \leq 2 N\exp\left( - \frac{\varphi \  m_{\mathrm{min}}}{27} \right)\le {\frac{1}{2n^z}}.
\end{align}
In particular, this means that in the $\varphi$ region 
\begin{equation}
   \label{eq:mu-condition 4}
\varphi \;\ge\; 27\big(\log(2N)+\log(2n^z)\big) / m_{\mathrm{min}}
\end{equation}
we get that, with probability at least \(1-\frac{1}{2n^z}\), the following hold true
\begin{align}
 \mathbb{E}(\mu_b)-\frac{1}{3}\mathbb{E}(\mu_b)\le\mu_b\le \mathbb{E}(\mu_b)+ \frac{1}{3}\mathbb{E}(\mu_b)\\
\varphi  |\widetilde{\mathbf{\Phi}}_b|+\frac{1}{3}\varphi |\widetilde{\mathbf{\Phi}}_b|\le |\mathbf{\Phi}_b|\le\varphi  |\widetilde{\mathbf{\Phi}}_b|+\frac{1}{3}\varphi |\widetilde{\mathbf{\Phi}}_b|
\end{align}
which means that 
\begin{equation}
\label{eq: lem, max phi}
    \max_{b \in [N]} |\mathbf{\Phi}_b|\le \frac{4}{3}\varphi\max_{b \in [N]} |\widetilde{\mathbf{\Phi}}_b|.
\end{equation}

Now, to proceed, we must also analyze the behavior of \(|\mathbf{X}|\). We first know that
\[
|\mathbf{X}|\sim \mathrm{Binomial}(\sum_{k=1}^{d}{n \choose k},\varphi)
\]
with 
\begin{equation}
\label{eq: lem, card X}
    \mathbb{E}(|\mathbf{X}|)=\sum_{\mathbf{e}\in \bigcup_{k=1}^{d}\mathbf{A}_{n,k}} \varphi_{\mathbf{e}}=\sum_{\mathbf{e}\in \bigcup_{k=1}^{d}\mathbf{A}_{n,k}} \varphi=\varphi \sum_{k=1}^{d}{n \choose k}.
\end{equation}
Let us recall from the Chernoff bound \cite[Theorem~4.5]{mitzenmacher2017probability}, that for every \(\varepsilon\in(0,1)\), we have
\begin{equation} \Prob\left(|\mathbf{X}|\le (1-\varepsilon)\mathbb{E}[|\mathbf{X}|]\right)\le\exp\left(-\frac{\varepsilon^2\mathbb{E}[|\mathbf{X}|]}{2}\right)   
\end{equation}
which, after choosing \(\varepsilon=\frac{1}{5}\), yields
\begin{equation}
 \Prob\left(
|\mathbf{X}|\le \frac{4}{5}\mathbb{E}[|\mathbf{X}|]
\right)
\le
\exp\left(
-\frac{\mathbb{E}[|\mathbf{X}|]}{50}
\right).   
\end{equation}

Following very similar steps to those between and 
\eqref{eq: lem 11, 5} and \eqref{eq:mu-condition 4b}, we can readily see that for 
\begin{equation}
    \label{eq: lem, phi_delta,1}
    \varphi\ge\frac{50\log(2n^z)}{\sum_{k=1}^{d}{n \choose k}}
\end{equation}
then
\begin{equation}
  \exp\left(-\frac{\varphi \sum_{k=1}^{d}{n \choose k}}{50}\right)\le\exp(-\log(2n^z))=\frac{1}{2n^z}
\end{equation}
which in turn means that, when we are in the region $\varphi\ge\frac{50\log(2n^z)}{\sum_{k=1}^{d}{n \choose k}}$ in \eqref{eq: lem, phi_delta,1}, then with probability at least \(1-\frac{1}{2n^z}\), it is the case that
\begin{equation}
\label{eq: lem, lowerbound on X}
  |\mathbf{X}|\ge\frac{4}{5}\varphi \sum_{k=1}^{d}{n \choose k}. 
\end{equation}
We are interested in understanding the region of $\varphi$ that lies in the intersection of the regions of \eqref{eq:mu-condition 4}  and \eqref{eq: lem, phi_delta,1} (because we need both \eqref{eq:mu-condition 4} and \eqref{eq: lem, phi_delta,1} to simultaneously hold true). As we will see in Appendix~\ref{appendix: Discussion u}, under the additional assumption that $d\leq \frac{n}{16}$, the region
\begin{equation}\label{eq:varphi_mindelta_n2}\varphi \geq \varphi_{\mathrm{min},\delta}^{n}\triangleq \min \left\{\frac{54N\log(4Nn^z)}{\sum_{k=1}^{d}\binom{n}{k} - (2^{2d+1})N},1\right\}\end{equation}
lies in the desired intersection of \eqref{eq:mu-condition 4}  and \eqref{eq: lem, phi_delta,1}. 

We can thus now conclude that under the assumption of $\varphi \geq \varphi_{\mathrm{min},\delta}^{n}$ and $d\leq \frac{n}{16}$ as well as $N\leq \binom{\lfloor\sqrt{\frac{nd}{2}}\rfloor}{d}$, the IC design guarantees  
\begin{equation} \label{eq:deltaXB1}
\delta_{\mathbf{X}}=\frac{\max_{b \in [N]} |\mathbf{\Phi}_b|}{\lceil|\mathbf{X}|\rceil}\le \frac{ \frac{4}{3}\varphi \max_{b \in [N]} |\widetilde{\mathbf{\Phi}}_b|}{\lceil\frac{4}{5}\varphi \sum_{k=1}^{d}\binom{n}{k}/N\rceil}\le \frac{ \frac{4}{3}\varphi\max_{b \in [N]} |\widetilde{\mathbf{\Phi}}_b|}{\frac{4}{5}\varphi\lceil \sum_{k=1}^{d}\binom{n}{k}/N\rceil}=\frac{5}{3} \delta\le 5
 \end{equation}
with probability at least \((1-\frac{1}{2n^z})^{2}\ge 1-\frac1{n^z}\), which completes the proof for general non-uniform hypergraphs.

We now provide the same proof, specialized to the case of uniform hypergraphs, where similar steps go through with three minor
modifications. 
First, the parameter \(m_{\mathrm{min}}\) now changes  from
\(
m_{\mathrm{min}}
=
\min_{b \in [N]} \sum_{k=1}^{d}
\left|\widetilde{\mathbf{\Phi}}_{b}^{k}\right|,
\) (please recall~\eqref{eq: random, m_min, m_max} as well as \eqref{eq: phi_b:1tod}) to 
\begin{equation}\label{eq:m-minA}
m_{\mathrm{min}}
=
\min_{b \in [N]} \left|\widetilde{\mathbf{\Phi}}_{b}^{d}\right|
\end{equation}
so as to reflect the $d$-uniform case. Second, while the condition in \eqref{eq:mu-condition 4b} remains the same $\varphi\  m_{\mathrm{min}}\;\ge\; 27\big(\log(2N)+\log(2n^z)\big)$, now  $m_{\mathrm{min}}$ takes the form in~\eqref{eq:m-minA}. Third, the $\varphi$ region in~\eqref{eq: lem, phi_delta,1} readily changes 
 -- after replacing $ \sum_{k=1}^{d}{n \choose k}$ with $ {n \choose d}$ -- to
\begin{equation}
    \label{eq: lem, phi_delta,1,u}
    \varphi\ge\frac{50\log(2n^z)}{{n \choose d}}.
\end{equation}
As before, we seek a range of $\varphi$ that lies in the intersection of the two regions $\varphi\;\ge\; \frac{27\big(\log(2N)+\log(2n^z)\big)}{m_{\mathrm{min}} }$ and \eqref{eq: lem, phi_delta,1,u}, and as before, as we will see in Appendix~\ref{appendix: Discussion u}, this region will be defined by 
\begin{equation}\label{eq:varphi_mindelta_u}
\varphi \geq \varphi_{\mathrm{min},\delta}^{u}
\triangleq \min\left\{ \frac{54N\log(4Nn^z)}{\binom{n}{d} - (2^{d+1})N},1\right\}.
\end{equation}
As before, we can again conclude that under the assumption of $\varphi \geq \varphi_{\mathrm{min},\delta}^{u}$ and $d\leq \frac{n}{4}$ (see again Appendix \ref{appendix: Discussion u}) as well as $N\leq \binom{\lfloor\sqrt{\frac{nd}{2}}\rfloor}{d}$, the IC design again guarantees 
\[
\delta_{\mathbf{X}}=\frac{\max_{b \in [N]} |\mathbf{\Phi}_b|}{\lceil|\mathbf{X}|\rceil}\le \frac{ \frac{4}{3}\varphi \max_{b \in [N]} |\widetilde{\mathbf{\Phi}}_b|}{\lceil\frac{4}{5}\varphi \binom{n}{d}/N\rceil}\le \frac{ \frac{4}{3}\varphi\max_{b \in [N]} |\widetilde{\mathbf{\Phi}}_b|}{\frac{4}{5}\varphi\lceil \binom{n}{d}/N\rceil}=\frac{5}{3} \delta\le 5
 \]
with probability at least \((1-\frac{1}{2n^z})^{2}\ge 1-\frac1{n^z}\), which completes the proof for uniform hypergraphs as well.\hfill $\square$

\subsection{Proof of Lemma \ref{lem: phin,k, cardinality}}
\label{appendix: proof of phin,k, cardinality}
The proof follows the steps of the proof of {Lemma 1 in~\cite{IC}}, and is adapted here to include the case of non-uniform hypergraphs.
We begin the proof by first recalling from Section \ref{subsec: non uniform} that \(\widetilde{\mathbf{\Phi}}_\sigma^{k} = \bigcup_{\beta=\lceil \frac{d}{s}\rceil}^{d-1}\bigcup_{\mathcal{I}\subset \sigma} \mathbf{Q}_{\beta, \mathcal I, \sigma}^{k}\) (cf.\eqref{eq: phi_ k<d}). Let us also remember that for any 
\(\beta\in [\beta_{\mathrm{min}}^k, \beta_{\mathrm{max}}^k]\) (see \eqref{eq:beta_min} and \eqref{eq:beta_max}), each group \(\sigma \in {[f] \choose d}\) has exactly \({d \choose \beta}\) distinct subsets \(\mathcal I\) of cardinality \(\beta\), and that for each such \(\mathcal{I} \subset \sigma\), the cardinality of \(\mathbf{Q}_{\beta, \mathcal I, \sigma}^{k}\) is either \(q_{\beta}^{k}+1\) or \(q_{\beta}^{k}\) \textcolor{black}{where \(q_{\beta}^{k}= \lfloor \frac{t_\beta^{k}}{m_\beta}\rfloor\) (recall \( m_\beta\) from \eqref{eq: m_beta} and recall \(t_{\beta}^{k}=|\mathbf{Q}_{\beta,\mathcal{I}}^{k}|\))}.
Thus, we have 
    \(      \frac{t_\beta^{k}}{m_\beta}-1\le q_{\beta}^{k}= \lfloor \frac{t_\beta^{k}}{m_\beta}\rfloor\le \frac{t_\beta^{k}}{m_\beta}, 
    \)
which gives 
 \begin{equation}
   \label{eq: lem3, ub, 2}
 |\widetilde{\mathbf{\Phi}}_{\sigma}^{k}|\leq \sum_{\beta=\beta_{\mathrm{min}}^{k}}^{\beta_{\mathrm{max}}^{k}}(q_{\beta}^{k}+1)\binom{d}{\beta}\leq \sum_{\beta=\beta_{\mathrm{min}}^{k}}^{\beta_{\mathrm{max}}^{k}}\frac{t_\beta^{k}}{m_\beta}\binom{d}{\beta}+\sum_{\beta=\beta_{\mathrm{min}}^{k}}^{\beta_{\mathrm{max}}^{k}}\binom{d}{\beta}
\end{equation}
and
 \begin{equation}
   \label{eq: lem3, lb, 2}
 |\widetilde{\mathbf{\Phi}}_{\sigma}^k|\geq \sum_{\beta=\beta_{\mathrm{min}}^{k}}^{\beta_{\mathrm{max}}^{k}}q_{\beta}^{k}\binom{d}{\beta}\geq \sum_{\beta=\beta_{\mathrm{min}}^{k}}^{\beta_{\mathrm{max}}^{k}}\frac{t_\beta^{k}}{m_\beta}\binom{d}{\beta}-\sum_{\beta=\beta_{\mathrm{min}}^{k}}^{\beta_{\mathrm{max}}^{k}}\binom{d}{\beta}.
\end{equation}
Now using \eqref{eq: m_beta}, we get
    \begin{equation}
      \label{eq: lem3, ub, 3}
        \frac{t_\beta^{k}}{m_\beta} \cdot\binom{d}{\beta}= \frac{t_\beta^{k}}{{f-\beta \choose d-\beta}}\binom{d}{\beta}=t_\beta^{k} \ \frac{{f \choose \beta }}{{f\choose d}}  
    \end{equation}
noting also that  
\begin{equation}
\label{eq: lem3, ub, 4}
    \sum_{\beta=\beta_{\mathrm{min}}^{k}}^{\beta_{\mathrm{max}}^{k}}\binom{d}{\beta}\leq \sum_{\beta=0}^{k}\binom{d}{\beta}=2^d-d.
\end{equation}
At this point, using \eqref{eq: lem3, ub, 3} and \eqref{eq: lem3, ub, 4}, we bound \eqref{eq: lem3, ub, 2} and \eqref{eq: lem3, lb, 2}  as follows
$ |\widetilde{\mathbf{\Phi}}_{\sigma}^k|\leq  \sum_{\beta=\beta_{\mathrm{min}}^{k}}^{\beta_{\mathrm{max}}^{k}}\frac{t_\beta^{k} \binom{f}{\beta}}{\binom{f}{d}}+2^d-d,$ and $
      |\widetilde{\mathbf{\Phi}}_{\sigma}^k|\geq  \sum_{\beta=\beta_{\mathrm{min}}^{k}}^{\beta_{\mathrm{max}}^{k}}\frac{t_\beta^{k} \binom{f}{\beta}}{\binom{f}{d}}-2^d+d.$  Now, for each \(\beta\) and \(\mathcal{I}\in {[f]\choose \beta}\), we have \({f \choose \beta}\) disjoint and equal sized sets \(\mathbf{Q}_{\beta, \mathcal{I}}^{k}\), which means that \(|\mathbf{Q}_{\beta, \mathcal{I}}^{k}|\binom{f}{\beta}=t_\beta^{k} \binom{f}{\beta}=|\mathbf{Q}_{\beta}^{k}|,\) and that
    \(\sum_{\beta=\beta_{\mathrm{min}}^{k}}^{\beta_{\mathrm{max}}^{k}}t_\beta^{k} \binom{f}{\beta}= \sum_{\beta=\beta_{\mathrm{min}}^{k}}^{\beta_{\mathrm{max}}^{k}}|\mathbf{Q}_{\beta}^{k}|={n \choose k}\), which in turn means that 
       \(\frac{{n \choose k}}{{f \choose d} }-2^d+d\le |\widetilde{\mathbf{\Phi}}_{\sigma}^k|\leq\frac{{n \choose k}}{{f \choose d} }+2^d-d\), 
   which concludes the proof. \hfill $\square$

    \subsection{Proof of Lemma \ref{lem: delta, power set}}
    \label{appendix: proof of delta, power set}

 Let us recall from Lemma \ref{lem: phin,k, cardinality}, that for \(k \in [1, d]\), we have
\[
\max_{\sigma \in {[f] \choose d}}|\widetilde{\mathbf{\Phi}}_\sigma^{k}|
\le  \frac{\binom{n}{k}}{N'}+2^d-d.
\]
Recall also from  Section~\ref{subsec: N' to N} -- when extending the partition from \(N'\) to \(N\) groups -- that each group is divided into \(q\) or \(p\) parts (please see~\eqref{eq: q,p,r} for the definitions of \(q\), \(p\), as well as of the parameter \(r\) used below). 
Now since \(q \le p\), we can conclude that 
\begin{align}
\label{eq: lemma 6,0}
  \max_{b \in [N]}|\widetilde{\mathbf{\Phi}}_b^{k}|
&\le\lceil \frac{\frac{\binom{n}{k}}{N'} +2^d -d}{q}\rceil  \le\lceil \frac{\binom{n}{k}}{q N'}\rceil  +\lceil\frac{2^d}{q} -\frac{d}{q}\rceil
\nonumber \\
&
\le \frac{\binom{n}{k}}{q N'}+1  +\lceil\frac{2^d}{q}\rceil -\lfloor\frac{d}{q}\rfloor \le \frac{\binom{n}{k}}{N-r}+2^d  
\end{align}
where the second inequality follows from the fact that \(\lceil x +y\rceil\le \lceil x\rceil+\lceil y\rceil\), the third inequality from the fact that \(\lceil x \rceil\le x+1\) and \(\lceil x-y\rceil \le \lceil x\rceil-\lfloor y \rfloor\), and the last inequality from the fact that
\(1 \le q = \left\lfloor \frac{N}{N'} \right\rfloor \le d\) (see \cite[Lemma~5]{IC}) 
and from the fact that \(N - r = N' q\). 

Let us also recall from \eqref{eq: group sigma, power set} that
\begin{equation}
 \widetilde{\mathbf{\Phi}}_b=\bigcup _{k=1} ^{d}\widetilde{\mathbf{\Phi}}_b^{k}
\end{equation}
which means that  
\begin{equation}
 \max_{b \in [N]}|\widetilde{\mathbf{\Phi}}_b|\le\sum_{k=1} ^{d}\max_{b \in [N]}|\widetilde{\mathbf{\Phi}}_b^{k}|\le \frac{\sum_{k=1}^{d}\binom{n}{k}}{N-r}+\sum_{k=1} ^{d}2^d\le \frac{\sum_{k=1}^{d}\binom{n}{k}}{N-r}+2^{d}d
\end{equation}
which allows us, for the complete non-uniform hypergraph case, to bound \(\delta\) as follows
\begin{align}
   \label{eq: th1: delta}
   \delta= \frac{\max_{b \in [N]}|\widetilde{\mathbf{\Phi}}_b|}{\lceil\sum_{k=1}^{d}\binom{n}{k}/ N\rceil}\le\frac{\sum_{k=1}^{d}\max_{b \in [N]}|\widetilde{\mathbf{\Phi}}_b^{k}|}{\sum_{k=1}^{d}\binom{n}{k}/ N} &\le \frac{N}{N-r}+\frac{2^{d} d N}{\sum_{k=1}^{d}\binom{n}{k}}\\
   \label{eq: th1: delta2}
   &=\frac{N}{N' \lfloor\frac{N}{N'}\rfloor}+\frac{2^{d} d N}{\sum_{k=1}^{d}\binom{n}{k}}\\
   \label{eq: th1: delta3}
 &\le 2+\frac{2^{d} d N}{\sum_{k=1}^{d}\binom{n}{k}}
\end{align}
where the transition from~\eqref{eq: th1: delta} to \eqref{eq: th1: delta2} uses \eqref{eq: lemma 6,0} and also uses that
\(N = N'  q + r\) where \(q = \left\lfloor \frac{N}{N'} \right\rfloor\), while the last inequality holds because  
\(f(x) = \frac{x}{\lfloor x \rfloor} \le 2\) for \(x=\frac{N}{N'} \ge 1\).

Finally, applying our constraints of \(N\le {\lfloor\sqrt{\frac{nd}{2}}\rfloor \choose d}\) and \(n\ge 5d\), and after  substituting \(\sum_{k=1}^{d}\binom{n}{k}\ge {n \choose d}\), all in \eqref{eq: th1: delta3}, we get the desired
\textcolor{black}{\begin{align}
\label{eq:th1: delta4}
\delta \; \le\; 2^d d \frac{ {\lfloor\sqrt{\frac{nd}{2}}\rfloor \choose d}}{\binom{n}{d}}+2 \le 2^d d\frac{ (\lfloor\sqrt{\frac{nd}{2}}\rfloor)(\lfloor\sqrt{\frac{nd}{2}}\rfloor-1)\cdots(\lfloor\sqrt{\frac{nd}{2}}\rfloor-d+1)}{n(n-1)\cdots (n-d+1)}+2\\
\;\le\; 2^{d} d (\frac{\sqrt{\frac{nd}{2}}}{n})^{d}+2
\;\le\; 2^{d} d (\sqrt{\frac{d}{2n}})^{d}+2\le  \frac{d}{(2.5)^{\frac{d}{2}}}+2\le 3.
\end{align}}
Similar steps follow for the \(d\)-uniform case, where first, as we recall from Lemma \ref{lem: delta, power set}, we have
\begin{equation}
    \max\limits_{\sigma \in \binom{[f]}{d}}  |\widetilde{\mathbf{\Phi}}_{\sigma}|=\max\limits_{\sigma \in \binom{[f]}{d}} |\widetilde{\mathbf{\Phi}}_{\sigma}^{d}| \leq\frac{{n \choose d}}{N'}+2^{d}-d
\end{equation}
and after transitioning from \(N'\) to \(N\) groups, we have
\begin{equation}
\label{eq: th1, maxmin, Sett2}
    \max_{b \in [N]}|\widetilde{\mathbf{\Phi}}_b|= \max_{b \in [N]}|\widetilde{\mathbf{\Phi}}_b^{d}|\leq \lceil\frac{\frac{{n \choose d}}{N'}+2^{d}-d}{q}\rceil \le \frac{{n \choose d}}{N-r}+2^{d}
\end{equation}
which now yields 
\begin{align}
   \label{eq: th1: deltau}
   \delta= \frac{\max_{b \in [N]}|\widetilde{\mathbf{\Phi}}_b|}{\lceil\binom{n}{d}/ N\rceil}\le\frac{\max_{b \in [N]}|\widetilde{\mathbf{\Phi}}_b^{d}|}{\binom{n}{d}/ N} &\le \frac{N}{N-r}+\frac{ 2^{d} N}{\sum_{k=1}^{d}\binom{n}{k}}\\
   \label{eq: th1: delta2u}
   &=\frac{N}{N' \lfloor\frac{N}{N'}\rfloor}+\frac{2^{d} N}{\binom{n}{d}} \le 2+\frac{2^{d} N}{\binom{n}{d}}
\end{align}
which, after applying \(N\le {\lfloor\sqrt{\frac{nd}{2}}\rfloor \choose d}\) and \(n\ge 2d\), yields the desired
\( \delta\le 2+ 2^d(\frac{\sqrt{\frac{nd}{2}}}{n})^d\le 3.
\) \hfill $\square$

\textcolor{black}{\subsection{Proof of Corollary \ref{cor:fasterZ-cardinality}}}
\label{appendix: Proof of cor, card X}
In terms of the converse, we recall that the edge set \(\mathbf X\) is sampled according to the random model, where each edge \(\mathbf e\in \Adn\) is retained with probability satisfying
\(
\varphi_{\mathbf e}\ge \varphi_{\mathrm{min}}.
\) (cf.~\eqref{eq:phi1})
For the uniform case, let us quickly then note that 
\begin{equation}
\mathbb{E}[|\mathbf{X}|]=\sum_{\mathbf e\in \Adn}{\varphi_{\mathbf{e}}\ge\varphi_{\mathrm{min}}{n \choose d} = 15nN\left(\log N+\frac{4}{5}\frac{\log (2n^z)}{n}\right). }
\end{equation}
and let us define 
\begin{equation}
\mathcal U\triangleq\left\{\mathbf{X} \ : \ |\mathbf X|\ge7.5nN\left(\log N+\frac{4}{5}\frac{\log (2n^z)}{n}\right)\right\} = \left\{ \mathbf{X} \ : \ |\mathbf X|\ge \frac{\varphi_{\mathrm{min}}|\mathbf{A}|}{2}\right\}.
\end{equation}
Now we employ the machinery of Lemma~\ref{lem:prob_pi}, and recall~\eqref{eq: lem1,eq1}, which, after setting \(t=\frac{n}{2N^{1/d}}\) and  \textcolor{black}{using \(E_{\mathrm{th}}=\varphi_{\mathrm{min}}\frac{{n \choose d}}{2N}\), }gives us
\begin{equation}
    \Prob\left[ \{  \pi_{\mathbf{X}}^{\star}<t\}\cap\{|\mathbf{X}|\ge N E_{\mathrm{th}}\}\right]\le \Prob\left[\exists \mathcal{S} : |\mathcal{S}| = t, Y_\mathcal{S} \ge E_{\mathrm{th}}\right]
\end{equation}
\textcolor{black}{where now $Y_\mathcal{S}$ is a slightly modified random variable} \(
Y_\mathcal{S} \sim \mathrm{Binomial}\left(\binom{t}{d}, \varphi_{\mathrm{min}}\right).
\) 
\textcolor{black}{At this point, everything follows as in the proof of Lemma~\ref{lem:prob_pi}, and we can now readily see that after applying our own value of $E_{\mathrm{th}}$ in~\eqref{eq: lem1, union prob,1}, we get}
\textcolor{black}{\begin{align}
\label{eq: A|B1}
       \Prob\left(\{\pi_{\mathbf{X}}^{\star}<t\}  \mid  \ \mathbf{X} \in \mathcal U\right) =\frac{\Prob\left(\{\pi_{\mathbf{X}}^{\star}<t\}  \cap  \ \mathbf{X} \in \mathcal U\right)}{\Prob\left(\mathbf{X} \in \mathcal U\right)} \\
       \label{eq: A|B2}
       \le \frac{\frac{1}{2n^z}}{1-\frac{1}{n^z}}\\
       \label{eq: A|B3}
       \le \frac{2}{3n^z}
\end{align}
The step \eqref{eq: A|B1} to \eqref{eq: A|B2} follows from \eqref{eq: x in u} and \eqref{eq: lem1, union prob,1}. The \eqref{eq: A|B3} is derived by substituting \(n\ge 2d=4\) and \(z=1\) in \eqref{eq: A|B2}.}
This concludes the converse part, which says that given \(|\mathbf X|\ge7.5nN(\log N + \frac{4}{5}\frac{\log(2 n^z)}{n}) \), then with probability at least \(1-\frac{2}{3n^z}\), no scheme can achieve \(\pi_{\mathbf X}<t\). 
The achievability part of the proof is identical to that of Theorem~\ref{thm:main}, and it follows directly from the guarantees of the IC design on \(\pi_{\mathbf X}\). \textcolor{black}{The same proof captures the non-uniform case as well for its corresponding \(\mathcal U=\left\{\mathbf{X} \ : \ |\mathbf X|\ge6.5nN\left(\log N+\frac{12}{13}\frac{\log (2n^z)}{n}\right)\right\}\) . }\hfill $\square$

\subsection{Statement and Proof of Lemma~\ref{lem:K1} and Lemma~\ref{lem:K2}}
\label{appendix:K2} 

\begin{lemma}[Simplification of expression] \label{lem:K1}
For all \(d\ge 2\) and \(N\ge 2\), then 
\[
\frac{6N ^{1-\frac{1}{d}}\log\left((2\sqrt{2}e)^{d} N\right)}{{ \sqrt{2}d}}
\leq
13 N \log N .
\]
\end{lemma}

\begin{proof}
We seek to prove that 
\(
\frac{6 N ^{1-\frac{1}{d}}\log\left((2\sqrt{2}e)^{d} N\right)}{{ \sqrt{2}d}}\leq13 N \log N ,
\)
and after dividing by \(N^{1-\frac1d}>0\), to prove that
\(
\frac{6}{\sqrt{2}}\log(2\sqrt{2}e)+\frac{6}{\sqrt{2}d}\log N\leq13N^{1/d}\log N .
\)
or equivalently, to prove that
\begin{equation} \label{eq:abcd1}
\left(13N^{1/d}-\frac{6}{\sqrt{2}d}\right)\log N
\ge
\frac{6}{\sqrt{2}}\log(2\sqrt{2}e).
\end{equation}
Since \(N\ge 2\), it is the case that \( \log N\ge \log 2, \ \ \text{and} \ \ N^{1/d}\ge 2^{1/d}=e^{(\log 2)/d}\ge1+\frac{\log 2}{d},\)
which we then apply to the LHS of \eqref{eq:abcd1}, to get
\begin{multline}
\left(13N^{1/d}-\frac{6}{\sqrt{2}d}\right)\log N
\ge
\left(13\left(1+\frac{\log 2}{d}\right)-\frac{6}{\sqrt{2}d}\right)\log 2
\\ =
13\log 2
+
\left(13\log 2-\frac{6}{\sqrt{2}}\right)\frac{\log 2}{d}
\geq
13\log 2
\geq
\frac{6}{\sqrt{2}}\log(2\sqrt{2}e),
\end{multline}
which proves the lemma.
\end{proof}

\begin{lemma}[Simplification of expression] \label{lem:K2}
For all \(d\ge 2\) and \(N\ge 2\), then 
\[
15N\log N
>
\frac{6N^{1-\frac1d}\log\!\left((2e)^dN\right)}{d}.
\]
\end{lemma}

\begin{proof}
We seek to prove that \(15N\log N>\frac{6N^{1-\frac1d}\log\!\left((2e)^dN\right)}{d}\). After dividing by \(N^{1-\frac1d}>0\), it suffices to prove that \(15N^{1/d}\log N>6\log(2e)+\frac{6}{d}\log N\),
or equivalently, to prove that
\begin{equation} \label{eq:abcd2}
\left(\frac{5d}{2}N^{1/d}-1\right)\log N>d\log(2e).
\end{equation}
Similar to Lemma \ref{lem:K1}, we have \(\log N\ge \log 2\) and \(N^{1/d}\ge 1+\frac{\log 2}{d}\),
which we then apply to the LHS of \eqref{eq:abcd2}, to get
\begin{multline}
\left(\frac{5d}{2}N^{1/d}-1\right)\log N\ge\left(\frac{5d}{2}\left(1+\frac{\log 2}{d}\right)-1\right)\log 2
\\
=\frac{5d\log 2}{2}+\left(\frac{5\log 2}{2}-1\right)\log 2\geq d\log(2e) +\left(\frac{5\log 2}{2}-1\right)\log 2\geq
d\log(2e)
\end{multline}
where we used \(\frac{5}{2}\log 2>\log(2e)\) and \(\frac{5}{2}\log 2>1\). This proves the lemma.
\end{proof}

\section{Additional Proofs} \label{appendix:AdditionalProofs}
\subsection{Proof of Equation~\eqref{eq:pi2sd}}
\label{appendix: bound N}
We here show that for any \(n\), \(d\leq \frac{n}{2}\), and \(N \le {\left\lfloor \sqrt{\frac{nd}{2}} \right\rfloor \choose d}\), and with \(s = \left\lfloor \frac{n}{f + d} \right\rfloor + 1\), the vertex footprint achieved by the IC design, satisfies \(\pi \le 2sd\).

Recall 
(cf.~\eqref{eq: f def}) parameter \(f = \max \left\{ r \in \mathbb{Z}^{+} \,\mid\, \binom{r}{d} \leq N \right\}\), and let us consider the case of interest here, where \(f \nmid n\). Let { \(g = n - sf\). From \cite[Lemma~3]{IC}, we have \(\pi = sd + g\). Therefore, it suffices to ensure that \(0 \le g \le sd\), which would imply that \(\pi \le 2sd\). Towards this, let us define  
\(\mathcal{F}_{\mathrm{valid}} \triangleq \left\{f \in [d,n] \subset \mathbb{Z} \ \mid\  \left\lfloor\frac{n}{f + d}\right\rfloor + 1 \leq \left\lfloor\frac{n}{f} \right \rfloor \right\}, \)
and let us claim that if \(f \in \mathcal{F}_{\mathrm{valid}}\), then \(n\) can be written, as suggested above, as \(n = f  s + g\), for some \(g \in [0,s  d]\).} To prove the claim, suppose \(f \in \mathcal{F}_{\mathrm{valid}}\), which implies that  
\(s \leq \left\lfloor n/f \right\rfloor,\)  
and hence that 
\(g = n - s f \geq n - \left\lfloor \frac{n}{f} \right\rfloor f \geq 0.\) Further, we also have
\[
g = n - s f = n - \left( \left\lfloor \frac{n}{f + d} \right\rfloor + 1 \right) f
\leq n - \left\lceil \frac{n}{f + d} \right\rceil f
\leq n - \frac{n}{f + d} f
= \frac{n d}{f + d}.
\]

Since \(\frac{n}{f + d} < s\), it follows that
\begin{equation}\label{eq:gv1}
g \leq \frac{n d}{f + d} \leq s d.
\end{equation}
To complete the argument, it remains to ensure that \(f \in \mathcal{F}_{\mathrm{valid}}\). To ensure that the condition \(s \leq \left\lfloor n/f \right\rfloor\) holds, it is sufficient to guarantee that
\(
\frac{n}{f} - \frac{n}{f + d} \geq 1.
\) Rearranging the inequality leads to a quadratic inequality  
\(
f^2 + df - dn \leq 0.
\) Solving the quadratic inequality gives an upper bound on $f$ as follows
\begin{equation}
    \label{eq: f_max}
f \leq \frac{-d + \sqrt{d^2 + 4dn}}{2} \triangleq f_{\mathrm{max}}.
\end{equation}
Since \(\sqrt{d^2 + 4dn} \leq d + 2\sqrt{dn}\), we have
\(
f_{\mathrm{max}} \leq \sqrt{dn}
\) and let \(
f_{\mathrm{max}} =c \sqrt{dn}
\) for some $c\in [0,1]$. Let us now define $y\triangleq \frac{d}{n}$. Since $d\leq \frac{n}{2}$, we have $y\leq \frac{1}{2}$. From~\eqref{eq: f_max}, we obtain
\(
c\cdot\sqrt{dn}\le \frac{-d + \sqrt{d^2 + 4dn}}{2}.
\) Dividing both sides of this inequality with $\sqrt{dn}$ gives
\(
2c \le \sqrt{y+4} - \sqrt{y}.
\) Further, taking \(y = \frac{1}{2}\) yields
\(
2c \le \sqrt{\frac{9}{2}} - \sqrt{\frac{1}{2}} = \sqrt{2},
\)
and hence \(c \le \frac{1}{\sqrt{2}}\). We choose \(c = \frac{1}{\sqrt{2}}\), which gives \(f_{\mathrm{max}} = \sqrt{\frac{nd}{2}}\). {Thus, when \(n \ge 2d\), we have that, as long as \(
N\leq N_{\mathrm{max}} \triangleq {\left\lfloor \sqrt{\frac{nd}{2}} \right\rfloor \choose d}
\), then \(f \in \mathcal{F}_{\mathrm{valid}}\).}

Therefore, for given \(n\), \(d\leq \frac{n}{2}\), and \(N \le {\left\lfloor \sqrt{\frac{nd}{2}} \right\rfloor \choose d}\), and \(s = \left\lfloor \frac{n}{f + d} \right\rfloor + 1\), we have \(0 \le g \le sd\), and hence \(\pi\leq 2sd\).\hfill $\square$

\subsection{Proof of Equations~\eqref{eq:varphi_mindelta_n2} and \eqref{eq:varphi_mindelta_u}}\label{appendix: Discussion u}

We first handle the non-uniform case. We wish to show that under the assumption that $d\leq n/16$, the region defined in~\eqref{eq:varphi_mindelta_n2} lies in the intersection of the regions defined by \eqref{eq:mu-condition 4}  and \eqref{eq: lem, phi_delta,1}. 
It is indeed easy to first see that the region defined by~\eqref{eq:varphi_mindelta_n2}  lies inside the region defined by \eqref{eq: lem, phi_delta,1}, which states that \(
 \varphi\ge\frac{50\log(2n^z)}{\sum_{k=1}^{d}{n \choose k}}\), because \eqref{eq:varphi_mindelta_n2}, which states that \(
 \varphi \geq  \min \left(\frac{54N\log(4Nn^z)}{\sum_{k=1}^{d}\binom{n}{k} - (2^{2d+1})N},1\right)
 \),  has a visibly bigger numerator and smaller denominator, compared to \eqref{eq: lem, phi_delta,1}.

Now let us focus on \eqref{eq:mu-condition 4} which says that 
\[\varphi\;\ge\; 27\big(\log(2N)+\log(2n^z)\big) / m_{\mathrm{min}}. \] To make the comparison, we will bound $m_{\mathrm{min}}$, where $
m_{\mathrm{min}}= \min\limits_{b\in [N]} |\widetilde{\mathbf{\Phi}}_{b}|$ (cf.~\eqref{eq: random, m_min, m_max}). However, for every $b \in [N]$, we have $|\widetilde{\mathbf{\Phi}}_{b}|=\sum_{k=1}^d |\widetilde{\mathbf{\Phi}}_{b}^k|$. Therefore, we have $m_{\mathrm{\min}}= \min_{b\in [N]}\sum_{k=1}^d |\widetilde{\mathbf{\Phi}}_{b}^k|$. To bound $|\widetilde{\mathbf{\Phi}}_{b}^k|$ for each $b\in [N]$, we first look at the case where there are $N'=\binom{f}{d}$ groups, and later we account for the transition from $N'$ to $N$. In this first case where $N=N'=\binom{f}{d}$, we can see that from Lemma~\ref{lem: phin,k, cardinality}, we have     
\[
\min\limits_{\sigma \in \binom{[f]}{d}} |\widetilde{\mathbf{\Phi}}_{\sigma}^{k}| \geq\frac{{n \choose k}}{{f \choose d}}-2^{d}+d
\]
and then, after summing up the above bounds over all $1\le k \le d$, we obtain
\begin{align}   \min\limits_{\sigma \in \binom{[f]}{d}}|\widetilde{\mathbf{\Phi}}_\sigma|= \min\limits_{\sigma \in \binom{[f]}{d}}\sum_{k=1}^d   |\widetilde{\mathbf{\Phi}}_{b}^k|\ge \sum_{k=1}^{d} \big( \frac{\binom{n}{k}}{\binom{f}{d}} -2^{d} + d \big)  &= \frac{\sum_{k=1}^{d} \binom{n}{k}}{\binom{f}{d}} - \sum_{k=1}^{d} (2^{d} - d) \\
     \label{eq: sum bounnd1}
    &\ge \frac{\sum_{k=1}^{d} \binom{n}{k}}{\binom{f}{d}} -2^d d  +d^{2} \\
    \label{eq: sum bounnd2}
    &\ge \frac{\sum_{k=1}^{d} \binom{n}{k}}{\binom{f}{d}} - 2^{2d}+d^{2} 
\end{align}
where the step from \eqref{eq: sum bounnd1} to \eqref{eq: sum bounnd2} follows since \(d\le 2^{d}\) for any \(d\ge 1\).

At this point, let us briefly recall that in Section \ref{subsec: N' to N} (where we extend the partition, from \(N'\) to \(N\) groups), each group is divided into \(q\) or \(p\) parts, and let us recall  (cf.~\eqref{eq: q,p,r}) the definitions of \(q\), \(p\), and \(r\). Now, let us note that  since 
\(q \le p\), we have  
\begin{align}
\label{eq: lem,m_min,0}
   m_{\mathrm{min}}= \min\limits_{b\in [N]} |\widetilde{\mathbf{\Phi}}_{b}| \ge \lfloor\frac{\frac{\sum_{k=1}^{d}{n \choose k}}{N'}- 2^{2d}+d^{2} }{p}\rfloor\ge \frac{\sum_{k=1}^{d}{n \choose k}}{pN'}-\frac{2^{2d}}{p}+\frac{d^{2}}{p}-1\\
   \label{eq: lem,m_min,1}
   \ge \frac{\sum_{k=1}^{d}{n \choose k}}{N+N'}-\frac{2^{2d}}{p}+\frac{d^{2}}{p}-1\\
   \label{eq: lem,m_min,2}
   \ge \frac{\sum_{k=1}^{d}{n \choose k}}{2N}-2^{2d}
\end{align}
where the transition from~\eqref{eq: lem,m_min,0}  to \eqref{eq: lem,m_min,1} follows from the facts that a) if \(r=0\), then \(pN'=N\), and b) if \(r\neq0\), then \(p=q+1\) and thus \(N\le pN'=N+N'-r\le N+N'\le 2N\). Furthermore, in the above, the transition from~\eqref{eq: lem,m_min,1} to \eqref{eq: lem,m_min,2} follows from the fact that \(1\le p\le d+1\le d^{2}\).
At this point, let us substitute \eqref{eq: lem,m_min,2} into condition \eqref{eq:mu-condition 4}, to get a new threshold on $\varphi$, which takes the form
\begin{equation}\label{eq:varphi_mindelta_n}
\varphi\ge \frac{54N\log(4Nn^z)}{\sum_{k=1}^{d}\binom{n}{k} - (2^{2d+1})N}.
\end{equation}
Since \(\varphi\le 1\), we define the threshold as
\begin{equation}
\label{eq: varphi, min, delta, n}
    \varphi_{\mathrm{min},\delta}^{n}=\min\left\{ \frac{54N\log(4Nn^z)}{\sum_{k=1}^{d}\binom{n}{k} - (2^{2d+1})N}, 1\right\}.
\end{equation}
Consequently, we conclude that the $\varphi$ region \(\varphi\ge \varphi_{\mathrm{min},\delta}^{n}\) lies in the intersection of  \eqref{eq:mu-condition 4}  and \eqref{eq: lem, phi_delta,1}. The condition \(d\le n/16\) is used to ensure that the denominator of \eqref{eq: varphi, min, delta, n} remains positive because, firstly,  
\begin{align}
    \sum_{k=1}^{d}\binom{n}{k} - (2^{2d+1})N\ge \binom{n}{d} - 2^{2d+1} N\ge N(\frac{{n \choose d}}{N}-2^{2d+1})\\
    \ge N(\frac{{n \choose d}}{{\lfloor \sqrt{\frac{nd}{2}} \rfloor \choose d}}-2^{2d+1})\ge N ((\frac{2n}{d})^{\frac{d}{2}}-2^{2d+1})
\end{align}
and because, secondly, the last term  \((\frac{2n}{d})^{\frac{d}{2}}-2^{2d+1}\) is non-negative when \(n\ge 16d\).

Similar steps follow for the case of $d$-uniform hypergraphs, where $\mathbf{X} \subseteq \mathbf{A}_{n,d}$. Now, from Lemma \ref{lem: delta, power set}, we have
\begin{equation}
    |\widetilde{\mathbf{\Phi}}|=|\widetilde{\mathbf{\Phi}}_{\sigma}^d| \ge \frac{\binom{n}{d}}{\binom{f}{d}} - 2^{d} + d
\end{equation}
while with steps similar to those  in \eqref{eq: lem,m_min,0}, \eqref{eq: lem,m_min,1},
 \eqref{eq: lem,m_min,2}, we get
\begin{equation}
\label{eq: lem,m_min_u,0}
     m_{\mathrm{min}}\ge \frac{{n \choose d}}{2N}-2^{d}.
\end{equation}
This will yield a new threshold $\varphi$ of the form
\begin{equation}\label{eq:varphi_mindelta_uBB}
\varphi_{\mathrm{min},\delta}^{u}
=\min\left\{\frac{54N\log(4Nn^z)}{\binom{n}{d} - (2^{d+1})N},1\right\}
\end{equation}
and again the \(\varphi\) region \(\varphi\ge \varphi_{\mathrm{min},\delta}^{u}\) lies in the intersection of  \eqref{eq:mu-condition 4}  and \eqref{eq: lem, phi_delta,1,u}. Now we see that the condition \(d\leq n/4\) is sufficient to ensure that the denominator of \eqref{eq:varphi_mindelta_uBB} is positive, because
\begin{align}
    {n \choose d} - (2^{d+1})N\ge \binom{n}{d} - 2^{2d+1} N\ge N(\frac{{n \choose d}}{{\lfloor \sqrt{\frac{nd}{2}} \rfloor \choose d}}-2^{d+1})\ge N ((\frac{2n}{d})^{\frac{d}{2}}-2^{d+1})
\end{align}
where we see that  \((\frac{2n}{d})^{\frac{d}{2}}-2^{d+1}\ge 0\) holds for \(n\ge 4d\).\hfill $\square$

\end{appendices}
\printbibliography

\end{document}